\pdfoutput=1
\documentclass{tlp}
\usepackage{color}
\makeatletter
\newif\if@restonecol
\makeatother

\usepackage[ruled, vlined, linesnumbered]{algorithm2e}
\usepackage[toc,page]{appendix}

\let\proof\relax

\usepackage{amsthm}
\usepackage{amsmath,amssymb}
\usepackage{graphicx}
\usepackage{verbatim} 
\usepackage{epic}
\usepackage{url}
\usepackage{enumitem}

\theoremstyle{definition}
\newtheorem{example}{Example}
\newtheorem{definition}{Definition}
\newtheorem{theorem}{Theorem}
\newtheorem{lemme}{Lemma}

\def\asperix.{\texttt{ASPeRiX}}

\def\PLE.{\ensuremath{\Pi}}

\newcommand{\corpspos}[1]{\ensuremath{body^+(#1)}}
\newcommand{\corpsneg}[1]{\ensuremath{body^-(#1)}}

\newcommand{\head}[1]{\ensuremath{head(#1)}}

\def\ATOMES.{\ensuremath{\bf A}}
\def\HERBRANDUNIVERS.{\ensuremath{\mathcal{A}}}
\def\VARIABLES.{\ensuremath{\mathcal{V}}}
\def\CONSTANTS.{\ensuremath{\mathcal{CS}}}
\def\FUNCTIONS.{\ensuremath{\mathcal{FS}}}
\def\PREDICATES.{\ensuremath{\mathcal{PS}}}
\def\TERMS.{\ensuremath{\bf T}}
\def\NN.{\ensuremath{\mathbb{N}}}

\def\SURE.{\ensuremath{W}}
\def\DEFAUTS.{\ensuremath{D}}
\def\THEORIE.{\ensuremath{(\SURE.,\DEFAUTS.)}}
\def\EXTENSION.{\ensuremath{E}}
\def\PREREQUIS.{\ensuremath{\alpha}}
\def\JUSTIFICATION.{\ensuremath{\beta}}
\def\JUSTIFICATIONS.{\ensuremath{\JUSTIFICATION._1\dots\JUSTIFICATION._n}}
\def\CONSEQUENT.{\ensuremath{\gamma}}
\newcommand{\defaut}[3]{\ensuremath{\frac{#1 : #2}{#3}}}
\newcommand{\defautsgenerateurs}[3]{\ensuremath{DG(#1,#2,#3)}}
\def\DEFAUTSGENERATEURS.{\defautsgenerateurs{\SURE.}{\DEFAUTS.}{\EXTENSION.}}
\def\DEFAUT.{\defaut{\PREREQUIS.}{\JUSTIFICATIONS.}{\CONSEQUENT.}}

\newcommand{\lprule}{\ensuremath{c \leftarrow a_1,~\dots,~a_n,~not~b_1,~\dots,~not~b_m.}}
\newcommand{\partialinter}[2]{\ensuremath{\langle #1, #2 \rangle}}
\newcommand{\computation}[1]{\ensuremath{{\langle #1 \rangle}_{i=0}^{\infty}}}
\newcommand{\outtextrule}[2]{\ensuremath{#1 \leftarrow #2.}}
\newcommand{\intextrule}[2]{\ensuremath{(#1 \leftarrow #2.)}}
\newcommand{\outtextfact}[1]{\ensuremath{#1.}}
\newcommand{\intextfact}[1]{\ensuremath{(#1.)}}

\def\BODY.{\ensuremath{body}}
\def\HEAD.{\ensuremath{head}}
\def\lparse.{\texttt{Lparse}}
\def\gringo.{\texttt{Gringo}}
\def\assat.{\texttt{Assat}}
\def\clasp.{\texttt{Clasp}}
\def\clingo.{\texttt{Clingo}}
\def\clingcon.{\texttt{Clingcon}}
\def\iclingo.{\texttt{iClingo}}
\def\smodels.{\texttt{Smodels}}
\def\cmodels.{\texttt{Cmodels}}
\def\pbmodels.{\texttt{Pbmodels}}
\def\gasp.{\texttt{GASP}}
\def\wasp.{\texttt{WASP}}
\def\asperixv.{\texttt{ASPeRiX 0.2.5}}
\def\gringov.{\texttt{Gringo 3.0.5}}
\def\claspv.{\texttt{Clasp 1.3.10}}
\def\dlv.{{\small{\sf{DLV}}}}
\def\lparsev.{\texttt{Lparse 1.1.1}}
\def\smodelsv.{\texttt{Smodels 2.32}}
\def\dlvv.{\texttt{DLV Dec 16 2012}}
\def\gaspv.{\texttt{GASP (june 2009)}}
\def\dlvcomplex.{\texttt{DLV-complex}}
\def\omigav.{\texttt{OMiGA Dec 3 2012}}
\def\omiga.{\texttt{OMiGA}}
\def\IN.{\ensuremath{IN}}
\def\OUT.{\ensuremath{OUT}}
\def\MBT.{\ensuremath{MBT}}
\def\si.{\ensuremath{\leftarrow}}
\def\gammacheck.{\ensuremath{\gamma_{check}}}
\def\gammacho.{\ensuremath{\gamma_{cho}}}
\def\gammapro.{\ensuremath{\gamma_{pro}}}
\def\true.{\mbox{\bf true}}
\def\false.{\mbox{\bf false}}
\def\NULL.{\mbox{\bf NULL}}

\def\subpi.{\ensuremath{\sqsubseteq}}
\def\litterauxSCCCourant.{\ensuremath{litterauxSCCCourant}}
\def\contientInstanceMBT.{\ensuremath{contientInstanceMBT}}
\def\SCCSansInstanceMBT.{\ensuremath{SCCSansInstanceMBT}}
\def\solve.{\ensuremath{solve}}
\def\instantiateRule.{\ensuremath{instantiateRule}}
\begin{document}

\bibliographystyle{acmtrans}

\submitted{24 March 2014}
\revised{24 November 2014}
\accepted{25 February 2015}

\title[\asperix.]{\asperix., a First Order Forward Chaining Approach for Answer Set Computing \footnote{This work was supported by ANR (National Research Agency), project ASPIQ under the reference ANR-12-BS02-0003.}}
\author[Claire Lef\`{e}vre, Christopher B\'{e}atrix, Igor St\'{e}phan, Laurent Garcia]{Claire Lef{\`e}vre, Christopher B\'{e}atrix, Igor St\'{e}phan, Laurent Garcia\\ LERIA, University of Angers, \\2 Boulevard Lavoisier, \\49045 Angers Cedex 01, France\\ Email: \{claire,beatrix,stephan,garcia\}@info.univ-angers.fr}
\maketitle

\begin{abstract}
The natural way to use Answer Set Programming (ASP) to represent knowledge in Artificial Intelligence 
or to solve a combinatorial problem is to elaborate a first order logic program with default negation. 
In a preliminary step this program with variables is translated in an equivalent propositional one by a  first tool: the grounder. 
Then, the propositional program is given to a second tool: the solver. 
This last one computes (if they exist) one or many answer sets (stable models) of the program, each answer set 
encoding one solution of the initial problem.  
Until today, almost all ASP systems apply this two steps computation.

In this article, the project \asperix. is presented as a first order forward chaining approach for Answer Set Computing.
This project was amongst the first to introduce an approach of answer set computing that escapes the 
preliminary phase of rule instantiation by integrating it in the search process. 
The methodology applies a forward chaining of first order rules that are grounded on the fly by 
means of previously produced atoms.
Theoretical foundations of the approach are presented,
the main algorithms of the  ASP solver \asperix. are detailed
and some experiments and comparisons with existing systems are provided.

\end{abstract}

\noindent \emph{KEYWORDS}: Answer Set Programming, solver implementation, grounding on the fly, first order, forward chaining.\\
\vbox{\hrule width \hsize}

\section{Introduction}
\label{sec:intro}
Answer Set Programming (ASP) is a very convenient paradigm to
represent knowledge in Artificial Intelligence (AI) and to encode
combinatorial problems~\cite{baral03,niemela99}. It has its roots in
nonmonotonic reasoning and logic programming and has led to a lot of
works since the seminal paper~\cite{gellif88b}. 
Beyond its ability to formalize various problems from AI or to encode combinatorial problems, ASP provides also an interesting way to practically solve such problems since some efficient
solvers are available.
In few words, if someone wants to use ASP to solve a problem, he has to write a logic
program in term of rules in a purely declarative manner in such a way that the answer sets (initially called
stable models in \cite{gellif88b}) of the program represent the solutions of his original problem.

\paragraph{\bf Illustration of ASP formalism}
Let us take two typical examples for which ASP
is suitable: the first example is devoted to knowledge representation in Artificial
Intelligence and the second one is a combinatorial problem.

\paragraph{KR problem}
This first example deals with default reasoning on  incomplete information.
It consists in describing knowledge
about birds.

\(
\begin{array}{l}
\outtextfact{bird(titi)}\\
\outtextfact{ostrich(lola)}\\
\outtextrule{bird(X)}{ostrich(X)}\\
\outtextrule{fly(X)}{bird(X), not~ostrich(X)}\\
\outtextrule{non\_fly(X)}{ostrich(X)}
\end{array}
\)

The meaning of the two first rules is that we have two objects:
\emph{titi} which is a bird and \emph{lola} which is an ostrich.
The meaning of the other rules is that an ostrich is a bird, a bird
which is not an ostrich flies and an ostrich does not fly.
Here, we are interested in deducing some properties about titi
and lola. Intuitively, we want that titi flies, lola is a bird and lola does not
fly. Concerning the information that lola does not fly, let us notice that it
is obtained by applying the last rule since lola is an ostrich and, then,
 the next to last rule cannot be applied in presence of
ostrich lola due to the part \emph{not} of this rule, called \emph{default negation}. 
Here, there is only one
answer set which contains all the deduced pieces of information: 
$\{bird(titi),  fly(titi), ostrich(lola), bird(lola), non\_fly(lola)\}$.


\paragraph{CSP problem}
The second example deals with the representation of a combinatorial problem:
possibles worlds are represented by nonmonotonic ``guess" rules 
and choice between these worlds is expressed by constraints.
The problem is then to find (at least) one solution corresponding to
a world verifying the constraints. This example is about graph
2-coloring.

\(
\begin{array}{l}
\outtextfact{vertex(1)}\\
\outtextfact{vertex(2)}\\
\outtextfact{edge(1,2)}\\
\outtextrule{red(X)}{vertex(X), not~blue(X)}\\
\outtextrule{blue(X)}{vertex(X), not~red(X)}\\
\outtextrule{}{red(X), red(Y), edge(X,Y)}\\
\outtextrule{}{blue(X), blue(Y), edge(X,Y)}
\end{array}
\)

This represents a graph with two vertices and
an edge between them (three first rules).
The two following rules are guess rules.
The fourth (resp. fifth) rule means that a vertex
which is not colored in blue (resp. red) has to
be colored in red (resp. blue).
The two last rules are constraints.
They mean that two adjacent vertices
can not have the same color. Here, we want to find
how the two vertices should be colored (knowing
that two colors are available). Intuitively, we have
two solutions: one with vertex 1 colored in blue
and vertex 2 colored in red and the other one with
vertex 1 colored in red and vertex 2 colored in
blue. This corresponds to the two answer sets of the
program: $\{vertex(1), vertex(2), edge(1,2), blue(1), red(2)\}$ and $\{vertex(1), vertex(2), edge(1,2), red(1), blue(2)\}$.
 However, let us note that, in this kind of problem,
we are often
interested in finding one solution rather than finding
all the possible solutions (and the determination of
only one answer set is enough).


As regards the form of the rules, we can notice that a program usually contains
different kind of rules.
The simplest ones are facts as \intextfact{bird(titi)} or \intextfact{vertex(1)}
representing data of the particular problem.
Some ones are about background knowledge as \intextrule{bird(X)}{ostrich(X)}.
Some others can be nonmonotonic as \intextrule{fly(X)}{bird(X), not~ostrich(X)}
for reasoning with incomplete knowledge.
In other cases, especially for combinatorial problems, nonmonotonic rules can be used to encode
alternative potential solutions of a problem as
\intextrule{red(X)}{vertex(X), not~blue(X)} and \intextrule{blue(X)}{vertex(X), not~red(X)}
expressing the two exclusive possibilities to color a vertex in a graph. 
Last, special headless rules are used to represent constraints 
of the problem to solve as \intextrule{}{red(X), red(Y), edge(X,Y)},
here, in order to not color with red two vertices linked by an edge.

With the examples above we can point out that knowledge representation in ASP 
is done by means of \emph{first order} rules.
But, from a theoretical point of view, answer set definition is  given for propositional programs
and the answer sets of a first order program are those of its ground 
instantiation with respect to its Herbrand universe (i.e. without variables). The first order
program has to be seen as an intensional version of the grounded propositional corresponding
program.

\vspace{.3cm}
\paragraph{\bf ASP systems}
Concerning the ASP sytems, their main goal is how to compute answer sets in an
efficient way. Let us recall some of their main features.
Until today, almost all systems available  to compute the
answer sets of a program follow the architecture described in
Fig.~\ref{fig:archi}.

\begin{figure}[h]
  \centering
  \includegraphics[height=4cm]{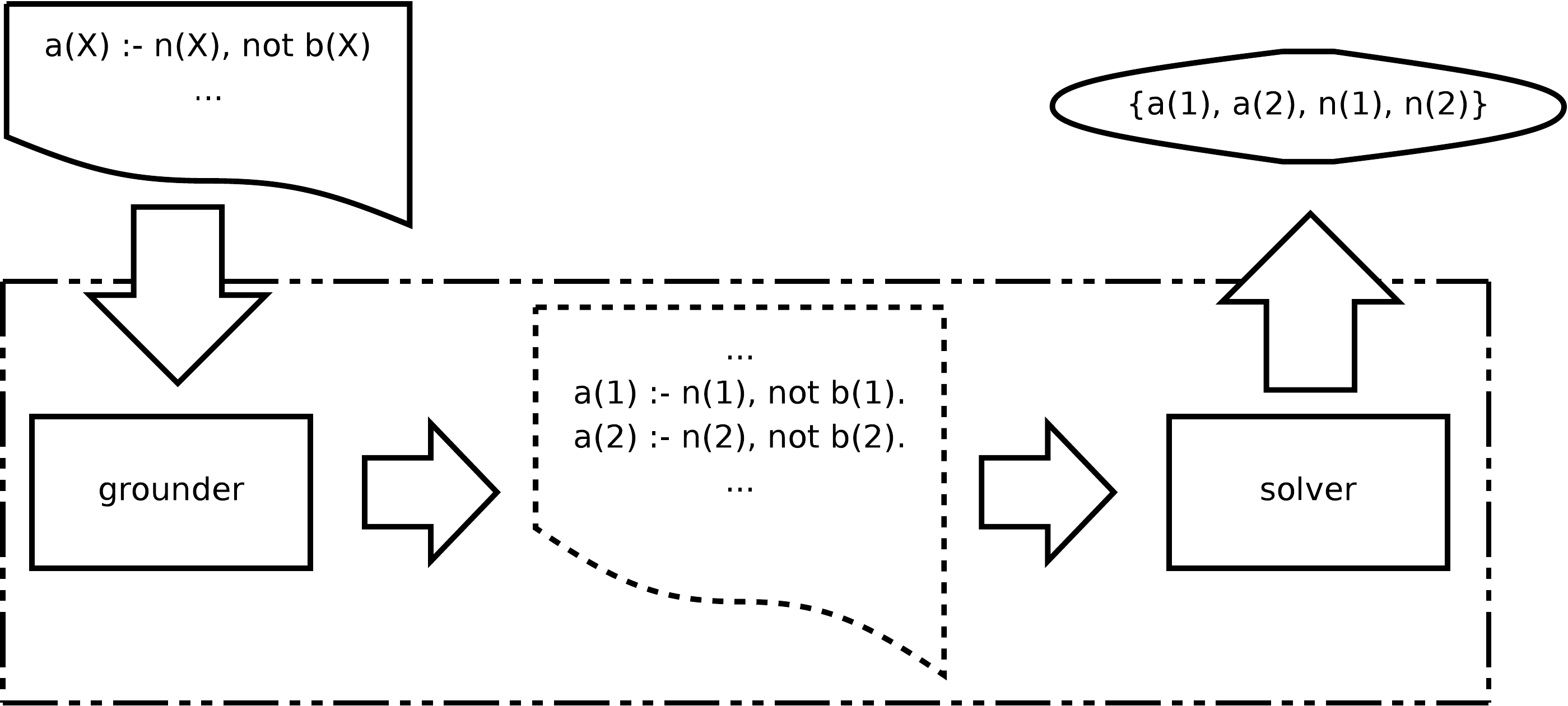}
  \caption{Architecture of answer set computation}
  \label{fig:archi}
\end{figure}

An ASP system begins its work by an instantiation 
phase in order to obtain a propositional program (and, as said above, the answer sets of
the first order program will be those of its ground instantiation). After this first \emph{grounding}
phase realized by a \emph{grounder} the solver starts the real phase of answer set computation 
by dealing with a finite, but sometimes huge, propositional program.
The main goal of each grounding system is to generate all propositional 
rules that can be relevant for a solver and only these ones,
while preserving answer sets of the original program.  
Current intelligent grounders simplify rules as much as possible.
Simplifications can lead to compute the unique answer set of some programs (for instance, programs that does not contain default negation) but it is no longer possible once the problem is combinatorial.
Anyway,
the grounding phase is firstly and fully processed before calling the solver.

For the grounder box we can cite
\lparse.~\cite{lparse} and \gringo.~\cite{gringo}, and for the
solver box \smodels.~\cite{smodels} and \clasp.~\cite{clasp}.
A particular family of solvers are
\assat.~\cite{assat}, \cmodels.~\cite{cmodels} and
\pbmodels.~\cite{pbmodels}, since they transform the answer set
computation problem into a (pseudo) boolean model computation problem
and use a (pseudo) SAT solver as an internal black box.  
In the system \dlv.~\cite{dlv}, symbolized in
Fig.~\ref{fig:archi} by the dash-line rectangle, the grounder
(\cite{cali08} describes a parallel version) is incorporated as an internal function. 
In the same way, \wasp.~\cite{wasp1} uses the \dlv. grounder~\cite{dlvgrounder}.

\vspace{.3cm}
\paragraph{\bf Grounding}
The main drawback of the preliminary grounding phase is that it may lead to a lot of
useless work as illustrated in the following examples.

The first examples illustrate the fact that the separation between the instantiation
phase and the computation phase can prevent the (efficient) use of  information relevant to the
computation.

\begin{example}
  \label{ex:exInut}
Let $P_{\ref{ex:exInut}a}$ be the following ASP program:
  \[
  P_{\ref{ex:exInut}a} = \left\{
    \begin{array}{lp{5mm}l}
      \outtextrule{a}{not~b},  \\
      \outtextrule{b}{not~a},  \\
      \outtextrule{}{a}, \\
      \outtextfact{p(0)},\\
      \outtextrule{p(X+1)}{a,~ p(X)}
    \end{array}
  \right\} 
  \]
Grounding of $P_{\ref{ex:exInut}a}$ is infinite  (if an upper bound for integers is not fixed) while it has a unique (and finite) answer set $\{b,~p(0)\}$.

Let $P_{\ref{ex:exInut}b}$ be the following ASP program:
  \[
  P_{\ref{ex:exInut}b} = \left\{
    \begin{array}{lp{5mm}l}
      \multicolumn{3}{l}{\outtextfact{p(1)},  \outtextfact{p(2)}, \dots,~\outtextfact{p(N)},}\\
      \outtextrule{a}{not~b},   && \outtextrule{aa(X,Y)}{pa(X),~pa(Y),~not~bb(X,Y)},\\
      \outtextrule{b}{not~a},   && \outtextrule{bb(X,Y)}{pb(X),~pb(Y),~not~cc(X,Y)},\\
      \outtextrule{pa(X)}{a,~ p(X)},&& \outtextrule{cc(X,Y)}{aa(X,Y),~ X<Y},\\
      \outtextrule{pb(X)}{b,~ p(X)}, && \outtextrule{}{a}
    \end{array}
  \right\} 
  \]

From the program $P_{\ref{ex:exInut}b}$, current grounders
generate roughly $2.5\times N^2$ rules.

\end{example}

In both programs, because of the constraint \intextrule{}{a} that eliminates from the
possible solutions every atom set containing $a$, it is easy to see
that rules \intextrule{p(X+1)}{a,~ p(X)} for $P_{\ref{ex:exInut}a}$ and \intextrule{pa(X)}{a,~ p(X)} for $P_{\ref{ex:exInut}b}$ are useless since they can never
contribute to generate an answer set of the corresponding program.
In  \(P_{\ref{ex:exInut}a}\) these useless rules are infinite \intextrule{p(X+1)}{a, p(X)} while 
they are ``only'' large in \(P_{\ref{ex:exInut}b}\): 
 $N$ rules with positive body containing $a$, like \intextrule{pa(1)}{a, p(1)}, and
then, the $N^2$ rules with $pa(X)$ in their positive body are useless
too.
In defense of the actual grounders, their inability to eliminate
these particular rules is not surprising since the reason justifying
this elimination is the consequence of a reasoning taking into account
the semantics of ASP. 
Thus, if we want to limit as much as possible the
number of rules and atoms to deal with, we have not to separate
grounding and answer set computing.

Example \(P_{\ref{ex:exInut}a}\) is a typical situation for planning problems where step $i+1$ must
be generated only if the goal is not reached at step $i$. Such situations 
are not tractable by grounders.  That is the reason why the number of steps needed
to reach the goal (or at least the maximum number of allowed steps)
is given as input of planning problems (in ASP competition for example). 
Yet it is rather counterintuitive having to know the step number to solve the problem before solving.

The next example illustrates that the grounding phase generates too much information regarding
the computation of one answer set.

\begin{example}
  \label{ex:3color}

  Let $P_{\ref{ex:3color}}$ be the program, as given in~\cite{niemela99}, encoding a 3-coloring
  problem on a $N$ vertices graph
  organized as a bicycle wheel (see below).
  $v$ stands for \emph{vertex}, $e$ for \emph{edge},
  $c$ for \emph{color}, $col$ for \emph{colored by}, $ncol$ for
  \emph{not colored by}.

  \(
  P_{\ref{ex:3color}} = \left\{
    \begin{array}{l}
      \outtextfact{v(1)}, \dots, \outtextfact{v(N)},~~\outtextfact{c(red)},~\outtextfact{c(blue)},~    \outtextfact{c(green)},\\
      \outtextfact{e(1,2)},  \dots, \outtextfact{e(1,N)},\\
      \outtextfact{e(2,3)}, \outtextfact{e(3,4)}, \dots,  \outtextfact{e(N,2)},\\

      \outtextrule{col(V,C)}{v(V),~ c(C),~not~ ncol(V,C)},\\
      \outtextrule{ncol(V,C)}{col(V,D),~c(C),~ C \neq D},\\
      \outtextrule{}{e(V,U),~ col(V,C),~col(U , C)}
  \end{array}
  \right\}
\)
\begin{minipage}{10mm}
 \includegraphics[width=1cm]{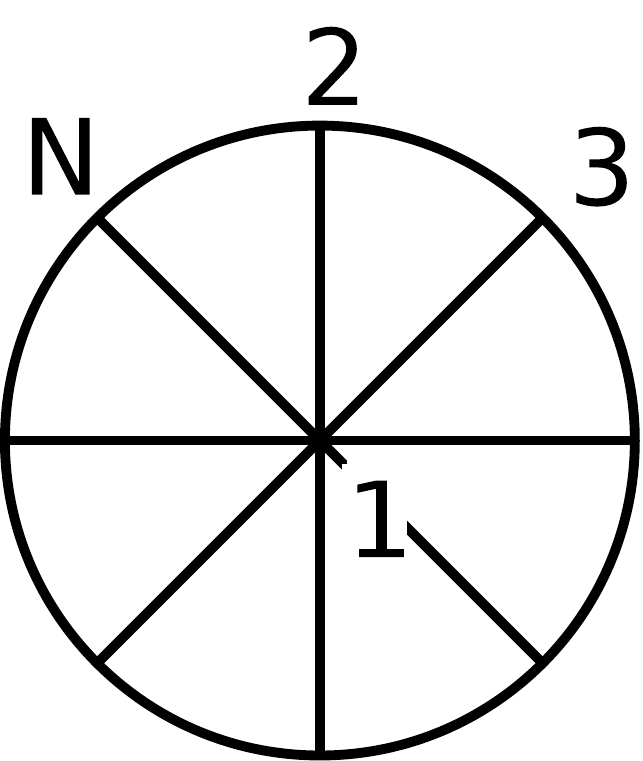}
  \end{minipage}.

From $P_{\ref{ex:3color}}$, current grounders
generate about $18 N$ rules.
If $N$ is even then $P_{\ref{ex:3color}}$ has no answer set and if $N$ is odd then it has 6 answer sets. 
\end{example}

Suppose that $P_{\ref{ex:3color}}$ has an answer set in which there is
$col(1, red)$.  Obviously, all the $N-1$ constraints like 
\intextrule{}{e(1, U),~ col(1, red),~col(U, red)} for all  \(U \in \{2,\dots,N\}\)
are necessary because they have to be checked. But, all the other
constraints like 
\intextrule{}{e(1, U),~ col(1, blue),~col(U, blue)},
and
\intextrule{}{e(1, U),~ col(1, green),~col(U, green)}
for all \( U \in \{2,\dots,N\}\)
can be considered as useless since vertex 1 is not colored by $blue$ or $green$.  
However, all these $2N-2$ constraints have been generated. So, the
time consumed by this task is clearly a lost time and the memory space
used by these data could have been saved.  Thus, if we are searching for
a single answer set, a lot of work would be done for nothing since the
grounded program contains the enumeration of all solutions when only
one is searched.

The last exemple shows that when the number of solutions is very important ASP
solvers have more difficulty to find one solution due to the grounding phase
generating a lot of information concerning all the solutions.

\begin{example}
\label{ex:hamcirccompgraph}

Let $P_{\ref{ex:hamcirccompgraph}}$ be the program, inspired from one given
in~\cite{niemela99}, encoding the
Hamiltonian cycle problem in a $N$ vertices complete oriented graph.
$v$ stands for \emph{vertex}, $a$ for \emph{arc},
$hc$ for \emph{in Hamiltonian cycle}, $nhc$ for \emph{not in  Hamiltonian cycle},
$s$ for \emph{start} and $r$ for \emph{reached}.
\[
P_{\ref{ex:hamcirccompgraph}}=\left\{
  \begin{array}{l}
    \outtextfact{s(1)},~~~~ \outtextfact{v(1)}, \dots, \outtextfact{v(N)},\hspace{10mm}
    \outtextrule{a(X,Y)}{v(X),~v(Y)},\\
    \outtextrule{hc(X,Y)}{s(X),~a(X,Y),~not~nhc(X,Y)},\\
    \outtextrule{hc(X,Y)}{r(X),~a(X,Y),~not~nhc(X,Y)},\\
    \outtextrule{nhc(X,Y)}{hc(X,Z),~a(X,Y),~ Y \neq Z},\\
    \outtextrule{nhc(X,Y)}{hc(Z,Y),~a(X,Y),~ X \neq Z},\\
    \outtextrule{r(Y)}{hc(X,Y)},\\
    \outtextrule{}{v(X),~not~r(X)}\\
  \end{array}
\right\}
\]
This program has $(N-1)!$ answer sets.
Whatever the number of desired solutions, current grounders generate about $2N^2$ rules with $hc$ predicate as head
and about $2N^3$ rules with $nhc$ predicate as head.
Thus, even if we restrict
our attention to the computation of one answer set,
all the ASP solvers preceded by a grounding phase
consume a huge amount of time when the graph has a few hundred vertices. 
\end{example}

This previous example illustrates another strange phenomenon. 
Sometimes, solving a trivial problem, as finding one Hamiltonian cycle in a
complete graph, is impossible for ASP systems. 
This is very counterintuitive since, in whole generality, in combinatorial 
problem solving the more solutions the problem has, the easier it is to find one of them. 
Again, the bottleneck for ASP systems seems to come from the huge number of rules
and atoms that are generated in first, delaying  and making the resolution more
difficult than it should be. 

Beyond these particular examples, the point to stress is that
grounders generate in extension all the search space (for all
potential solutions) that they give then to the solver. But, this
 is clearly not the approach of usual search algorithms.  A classical
coloring algorithm does not firstly enumerate, in extension, all
possible colorations for every vertex in the graph.  A finite domain solver
makes choices by instantiating some variables, propagates the
consequences of these choices, checks the constraints and by
backtracking explores its search space. Following this strategy it
instantiates and desinstantiates variables describing the problem to
solve all along its search process.  But, it does not build, a priori
and explicitly, all the possible tuples of variables and constraints
representing the problem to solve. That is why we think that if we want
to use ASP to solve very large problems we have to
realize the grounding process during the search process and not before
it. 

Is is important to notice that few works advocate the grounding of the program during the search of an answer set and not by a preprocessing.
Some aim at solving the grounding bottleneck by combining ASP to constraint programming: \cite{Baselice_Bonatti_Gelfond_ICLP_05} proposes to reduce the memory requirements for a very specific class of programs, i.e. multi-sorted logic programs with cardinality constraints, 
\cite{Balduccini_ASPOCP_09} proposes an algorithm to make cooperate an ASP solver and a Constraint Logic Programming solver in such a way that ASP is viewed as a specification language for constraint satisfaction problems and 
\cite{Ostrowski_Schaub_TPLP_12} describes the \clingcon. system which is a tight cooperation between the ASP solver \clasp. and the Constraint Programming solver \texttt{GeCode}.
The theory solvers (mainly arithmetic solvers) forbid instances that are in conflict with the constraints reducing by this way the size of the grounding image. 
Some others works use a forward chaining of rules that are instantiated as and when required: \gasp. \cite{gasp} and \asperix.~\cite{lefnic09a,lefnic09b}) developed at the same time, and more recently \omiga. \cite{omiga12}.
They are all based on the notion of computation given in~\cite{Liu2010295}.
\gasp. is implemented in Prolog and Constraint Logic Programming over finite domains.
Each rule instantiation and propagation is realized by building and solving a CSP.
\omiga. is implemented in Java and uses an underlying Rete network for instantiation and propagation.
\asperix., which is the one presented in this article, is implemented in C++.  
Instantiation and propagation are inspired by previous work realized on the \dlv. grounder
 which is based on the semi-naive evaluation technique of~\cite{Ullman}.

Last, concerning a direct handling of first order programs, let us note that there exists some 
works~\cite{gottlob96,eitlusu97,ferleelif07,linzhou07,Truszczynski_CR_12} dealing with first
order nonmonotonic logic programs. These works establish some relations between stable model
semantics and constraints systems or  second order logic or circumscription but they are not really
concerned by the explicit computation of answer sets.

The present paper is an extended version of~\cite{lefnic09a,lefnic09b}.
It details our approach of answer set computation that escapes the preliminary grounding
phase by integrating it in the search process and includes:
\begin{itemize}
\item theoretical foundations of the approach, ``mbt  \asperix. computation'', with complete proofs;
these computations are based on those of~\cite{Liu2010295} 
and include use of constraints and must-be-true propagation in order to guide the search;
\item a detailed description of the main algorithms;
\item experimentations of the resulting system, \asperix., and comparisons with other similar systems and other ``classical''
ASP systems.
Our methodology is particularly well suited for: 
\begin{itemize}
\item solving easy problems with a large grounding,
\item finding only one answer set for a program whose search space is large and proportional to the desired number
of  solutions,
\item solving problems for which pre-grounding is impossible because domains are infinite or open, or because some
pieces of knowledge come from outside (distributed systems for example).
\end{itemize}
\end{itemize}

The paper is organized as follows.
In Section~\ref{sec:backgrounds} we recall the theoretical backgrounds about ASP necessary to the
understanding of our work.
In Section~\ref{sec:asperix} we present our first order rule oriented approach of answer set computation and its implementation in the solver \asperix..
In Section~\ref{sec:experiments} experimental results are presented. 
We conclude in Section~\ref{sec:conclusion} by citing some new perspectives for ASP as a result of our innovative approach. Proofs of theorems are reported in~\ref{sec:proofs}.

\section{Theoretical Background}
\label{sec:backgrounds}
In this section, we give the main backgrounds of ASP framework useful to the understanding of this article.

Set \VARIABLES. denotes the infinite countable set of \emph{variables}, set \FUNCTIONS. denotes the set of \emph{function symbols}, set \CONSTANTS. denotes the set of \emph{constant symbols} and set \PREDICATES. denotes the set of \emph{predicate symbols}.
It is assumed that the sets \VARIABLES., \CONSTANTS., \FUNCTIONS. and  \PREDICATES. are disjoint and that the
set \CONSTANTS. is not empty.
Function $ar$ denotes the arity function from $\FUNCTIONS.$ to $\NN.^*$ and from $\PREDICATES.$ to $\NN.$  which associates to each function or predicate symbol its arity.
Set \TERMS. denotes the set of \emph{terms} defined by induction as follows: 
\begin{itemize}
\item if $v\in \VARIABLES.$ then $v\in \TERMS.$,
\item if $c\in \CONSTANTS.$ then $c \in \TERMS.$,
\item if $f\in \FUNCTIONS.$ with $ar(f)=n>0$ and $t_1,\dots,t_n \in \TERMS.$ then $f(t_1,\dots,t_n) \in \TERMS.$.
\end{itemize}
A \emph{ground term} is a term built over only the two last items of the previous definition.
The \emph{Herbrand universe} is the set of all ground terms.
Set \ATOMES. denotes the set of \emph{atoms} defined as follows:
\begin{itemize}
\item if $a\in \PREDICATES.$ with $ar(a)=0$ then $a \in \ATOMES.$,
\item if $p\in \PREDICATES.$ with $ar(p)=n>0$ and $t_1,\dots,t_n \in \TERMS.$ then $p(t_1,\dots,t_n) \in \ATOMES.$.
\end{itemize}
A \emph{ground atom} is an atom built over only ground terms.
The \emph{Herbrand base} denoted \HERBRANDUNIVERS. is the set of all ground atoms.

A \emph{normal logic program} (or simply \emph{program}) is a set of \emph{rules} like 
\begin{equation}
  \label{eq:lprule}
  \lprule\ ~~n\ge 0, m\ge 0
\end{equation}
where \(c, a_1,\dots,a_n, b_1,\dots,b_m \) 
are atoms.

The intuitive meaning of such a rule is: "if all the $a_i$'s are true and 
it may be assumed that all the $b_j$'s are false then one can
conclude that $c$ is true". 
Symbol $not$ denotes the \emph{default negation}.
A rule with no default negation is a \emph{definite rule} otherwise its is a \emph{nonmonotonic rule}.
A program with only definite rules is a \emph{definite logic program}.
A program is a \emph{propositional program} if all the predicate symbols are of arity 0. 

For each program $P$, we consider that 
the set $\CONSTANTS.$ (resp.  $\FUNCTIONS.$ and $\PREDICATES.$) consists of all constant (resp. function and predicate) symbols appearing in $P$.
These sets determine the set of ground terms and the set of ground atoms of the program.
A \emph{substitution} for a rule $r \in P$ is a mapping from the set of variables from $r$ to the set of ground terms of $P$.
A ground rule $r'$  is a \emph{ground instance} of a rule $r$ if there is a substitution $\theta$ for $r$ such that $r'=\theta(r)$, the rule obtained by substituting every variable in $r$ by the corresponding 
ground term in $\theta$.
The program $P$ (with variables) has to be seen as an intensional version of the 
program \(ground(P)\) defined as follows: given a rule $r$,
$ground(r)$ is the set of all ground instances of $r$
 and then,
\(ground(P)=\bigcup_{r \in P}ground(r)\).
Program $ground(P)$ may be considered as a propositional program.
Let us note that the use of function symbols leads to an infinite Herbrand universe, this point will be discussed in Section~\ref{subsec:discussion}.

\newpage
\begin{example}
\label{ex:ground_program}
 The program 

\(
      P_{\ref{ex:ground_program}}=\left\{
        \begin{array}{l}
          \outtextfact{n(1)}, \;  \outtextfact{n(2)}, \\
          \outtextrule{a(X)}{n(X),~not~b(X)},\\
          \outtextrule{b(X)}{n(X),~not~a(X)}\\
        \end{array}
      \right\}
\)

is a shorthand for the program

\(
      ground(P_{\ref{ex:ground_program}})=\left\{
        \begin{array}{l}
          \outtextfact{n(1)}, \;  \outtextfact{n(2)}, \\
          \outtextrule{a(1)}{n(1),~not~b(1)},\\
          \outtextrule{b(1)}{n(1),~not~a(1)},\\
          \outtextrule{a(2)}{n(2),~not~b(2)},\\
          \outtextrule{b(2)}{n(2),~not~a(2)}\\
        \end{array}
      \right\}
      \)

\end{example}

For a rule $r$
(or by extension for a rule set), we define:
\begin{itemize}
\item  \(\HEAD.(r)=c\) its \emph{head},
\item \(\corpspos{r}=\{a_1,\dots,a_n\}\) its \emph{positive body} and
\item \(\corpsneg{r}=\{b_1,\dots,b_m\}\) its \emph{negative body}.
\end{itemize}

The immediate consequence operator for a definite logic program $P$ is $T_P: 2^{\HERBRANDUNIVERS.} \rightarrow 2^{\HERBRANDUNIVERS.}$ such that $T_P(X)=\{\head{r} \mid r \in P,~ \corpspos{r} \subseteq X \}$.
The \emph{least Herbrand model} of $P$, denoted $Cn(P)$, is the smallest set of atoms closed under $P$, i.e., the smallest set $X$ such that $T_P(X) \subseteq X$.
It can be computed as the least fix-point of the consequence operator $T_P$.

The \emph{reduct} \(P^X\) of a normal logic program $P$ w.r.t. an atom set $X\subseteq \HERBRANDUNIVERS.$ is the definite logic program defined by:
\[
 P^X=\{\outtextrule{\HEAD.(r)}{\corpspos{r}} \mid r \in P,~ \corpsneg{r} \cap X = \emptyset\} 
\]
and it is the core of the definition of an \emph{answer set}.

\begin{definition}\cite{gellif88b}
  \label{def:sm}
 Let $P$ be a normal logic program and $X$ an atom set.  
$X$ is an answer set of $P$ if and only if \(X=Cn(P^X)\).
\end{definition}

For instance, the propositional program \(\{\outtextrule{a}{not~b}, \; \outtextrule{b}{not~a}\}\) has two answer sets \(\{a\}\) and \(\{b\}\).

\begin{example}
\label{ex:background}

  Taking again the program \(P_{\ref{ex:ground_program}}\),
  \(ground(P_{\ref{ex:ground_program}})\)
  has four answer sets:

\(
\begin{array} {cc}
   \{a(1), a(2), n(1), n(2)\},  
&  \{a(1), b(2), n(1), n(2)\},\\
  \{a(2), b(1), n(1), n(2)\},  
&  \{b(1), b(2), n(1), n(2)\}
\end{array}
\)

that are thus the answer sets of $P_{\ref{ex:ground_program}}$.
\end{example}

There is another definition of an anwer set for a normal logic program based on the notion of \emph{generating rules} which are the rules participating to the construction of the answer set.
These rules are important in our approach because they are exactly the rules fired in the \asperix. computation presented in the next section.

\begin{definition} \cite{nomore06}
\label{def:Generatingrules}
Let $P$ be a normal logic program and $X$ be an atom set.  
$GR_P(X)$, the set of \emph{generating rules} of $P$, is defined as 
$GR_P(X) = \{r \in P \mid \corpspos{r} \subseteq X \mbox{ and }  \BODY.^-(r) \cap X = \emptyset\}$.
\end{definition}

\begin{definition} \cite{nomore06}
\label{def:Groundedrules}
Let $R$ be a set of rules. $R$ is $grounded$ if there exists an enumeration
${\langle r_i \rangle}_{i \in [1..n]}$ of the rules of $R$ such that $ \forall
i \in [1..n], \BODY.^+(r_i ) \subseteq \HEAD.(\{r_j~|~j < i\})$.
\end{definition}

The next theorem is inspired by \cite{nomore06}. In \cite{nomore06}, X is an answer set of a program $P$
if and only if $X = Cn({GR_P(X)}^\emptyset)$. It can be reformulated by:

\begin{theorem} \cite{nomore06}
\label{the:AnswersetGR}
Let $P$ be a normal logic program and $X$ be an atom set.  
Then,  $X$ is an answer set of $P$ if and only if $X =  \HEAD.(GR_P(X))$
and $GR_P(X)$ is grounded.

\end{theorem}




Special headless rules, 
called \emph{constraints}, are admitted and considered equivalent to rules
like \intextrule{bug}{\dots,~ not~ bug} where $bug$ is a new symbol
appearing nowhere else. For instance, the program \(\{\outtextrule{a}{not~b},
\; \outtextrule{b}{not~a}, \outtextrule{}{a}\}\) has one, and only one, answer set
\(\{b\}\) because constraint \intextrule{}{a} prevents $a$ to be in an answer set.

When dealing with default negation, we call a \emph{literal} an atom, $a$, or the negation
of an atom, $not~a$. A literal $a$ is said to be \emph{positive}, and $not~a$ is said to be
\emph{negative}.
The corresponding atom $a$ of a literal $l$ is denoted by $at(l)$.
For a literal $l$ where $at(l)=a$, 
let us denote $pred(l)$ the function such that $pred(not\; a) = pred(a) = p$ with $p$ the predicate symbol of the atom $a$.

For purposes of knowledge representation, one may have to use conjointly
strong negation (like $\neg a$) and default negation (like $not~a$)
inside a same program. 
This is possible in ASP by means of an
\emph{extended logic program}~\cite{gellif91}
in which rules are built with \emph{classical} literals (i.e. an atom $a$ or its strong negation $\neg a$) instead of atoms only.
Semantics of extended logic programs distinguishes inconsistent answer sets from absence of answer set.
But, if we are not interested in inconsistent answer sets,
the semantics associated to an extended logic program is reducible to answer set
semantics for a normal logic program using constraints by taking into account the following
conventions:
\begin{itemize}
\item every classical literal $\neg x$ is encoded by the atom $nx$,
\item for every atom $x$, the constraint \intextrule{}{x, ~ nx} is added.
\end{itemize}
By this way, only consistent answer sets are kept. In this article, we do
not focus on strong negation and literal will never stand for classical literal.

Let us note that one can also use some particular atoms for (in)equalities and simple arithmetic
calculus on (positive and negative) integers.  
Arithmetic operations are treated as a functional arithmetic and comparison relations are treated as built-in predicates. 

Finally, a program $P$ is said to be \emph{stratified} iff there is a mapping $strat$ from \PREDICATES. to $\NN.$ such that, for each ground rule like (\ref{eq:lprule}), the two following conditions hold:
\begin{itemize}
\item $strat(pred(c)) \geq strat(pred(a_i))$ for all $i \in [1..n]$
\item $strat(pred(c)) > strat(pred(b_j))$ for all $j \in [1..m]$
\end{itemize}

\section{A First Order Forward Chaining Approach for Answer Set Computing}
\label{sec:asperix}
\subsection{\asperix. Computation}
\label{subsec:computation}
In this section, a characterization of answer sets for first-order normal logic programs, based on a concept of \emph{\asperix. computation}, is presented.
This concept is itself based on an abstract notion of \emph{computation}  for ground programs proposed in~\cite{Liu2010295}.
This computation  fundamentally uses a forward chaining of rules. 
It builds incrementally the answer set of the program and does not require the whole set of ground atoms from the beginning of the process.  
So, it is well suited to deal directly with first order rules by instantiating them during the computation.

The only syntactic restriction required by this methodology is that every rule of a program must be \emph{safe}. 
That is, all variables occurring in the head and all variables occurring in the negative body of a rule occur also in its positive body. 
Note that this condition is already required by all standard evaluation procedures.
Moreover, every constraint (i.e. headless rule) is considered given with the particular head \(\bot\) and is also safe.  
For the moment we do not consider function symbols but their use will be discussed in Section~\ref{subsec:discussion}.

An \emph{\asperix. computation} is defined as a process on a computation state based on a \emph{partial interpretation} which is defined as follows.

\begin{definition}
\label{def:partial_interpretation}
A \emph{partial interpretation} for a program $P$ is a pair $\partialinter{IN}{OUT}$ of disjoint atom sets included in the Herbrand base of $P$. 

\end{definition}

Intuitively, all atoms in $IN$ belong to a search answer set and all atoms in $OUT$ do not.

The notion of partial interpretation defines different status for rules.

\begin{definition}
\label{def:rulestatus}

  Let $r$ be a ground rule and $I = \langle IN, OUT\rangle$ be
  a partial interpretation. 
  \begin{itemize}
  \item $r$ is \emph{supported} w.r.t. $I$ when \(\BODY.^+(r) \subseteq IN\),
  \item $r$ is \emph{blocked} w.r.t. $I$ when \(\BODY.^-(r) \cap IN \neq \emptyset\),
  \item $r$ is \emph{unblocked} w.r.t. $I$ when \(\BODY.^-(r) \subseteq OUT\),
  \item $r$ is \emph{applicable} w.r.t. $I$ when $r$ is supported and
    not blocked.\footnote{The negation of blocked, \emph{not blocked},
      is different from \emph{unblocked}.}
  \end{itemize}
\end{definition}

An \asperix. computation is a forward chaining process that instantiates and fires one unique rule at each iteration according to two kinds of inference: 
a monotonic step of \emph{propagation} and a nonmonotonic step of \emph{choice}.
To fire a rule means to add the head of the rule in the set $IN$.

\begin{definition}
\label{def:Delta_Cho_Pro}
Let $P$ be a set of first order rules, $I$ be a partial interpretation and $R$ be a set of ground rules.

\begin{itemize}
\item \(\Delta_{pro}(P,I,R)\) is the set of all supported definite rules and supported unblocked nonmonotonic rules  from  \(ground(P) \setminus R \).
\item \(\Delta_{cho}(P,I,R)\) is the set of all applicable nonmonotonic rules  from \(ground(P) \setminus R \).
\end{itemize}
\end{definition}

It is important to notice that the two sets defined above, like the set $ground(P)$, do not need to be explicitly computed.
It is in accordance with the principal aim of this work that is to avoid their extensive construction. 
When necessary, a first-order rule of $P$ can be selected and grounded with propositional atoms occurring in $IN$ and $OUT$ in order to define a new (not already occurring in $R$) fully ground rule member of $\Delta_{pro}$ or $\Delta_{cho}$.
Because of the safety constraint on rules this full grounding is always possible.
These mechanisms are specified in more details in Subsection~\ref{subsec:gamma}. 
The sets $\Delta_{pro}$ and $\Delta_{cho}$ are used in the following definition of an \asperix. computation.
Specific case of constraints (rules with $\bot$ as head) is treated by adding $\bot$ into \OUT. set.
By this way, if a constraint is fired (violated), $\bot$ should be added into \IN.
and thus, $\langle IN, OUT \rangle$ would not be a partial interpretation.

\begin{definition}
\label{def:Asperixcomputation}

 Let $P$ be a first order normal logic program.
  An \emph{\asperix. computation} for $P$ is a sequence ${\langle R_i , I_i \rangle}_{i=0}^{\infty}$
  of ground rule sets $R_i$ and partial interpretations $I_i=\partialinter{IN_i}{OUT_i}$  that satisfies the following conditions:
  \begin{itemize}
  \item $R_0 = \emptyset$ and $I_0 = \partialinter{\emptyset}{\{\bot\}}$,
  \item (Revision) $\forall i \geq 1$, 
	\begin{itemize}
	\item[](Propagation) $R_i = R_{i-1} \cup \{r_i\}$ with $r_i \in \Delta_{pro}(P,I_{i-1}, R_{i-1})$\\
	and  $I_i = \partialinter{IN_{i-1} \cup \{\HEAD.(r_i)\}}{OUT_{i-1}}$
	\item[or](Rule choice) $\Delta_{pro}(P,I_{i-1}, R_{i-1}) = \emptyset$,\\
	$R_i = R_{i-1} \cup \{r_i\}$ with $r_i \in \Delta_{cho}(P,I_{i-1}, R_{i-1})$\\
	and $I_i = \partialinter{IN_{i-1} \cup \{\HEAD.(r_i)\}}{OUT_{i-1}\cup \BODY.^-(r_i)}$
	\item[or](Stability) $R_i = R_{i-1} $ and $I_i = I_{i-1}$,
	\end{itemize}
  \item (Convergence) $\exists i \geq 0,~\Delta_{cho}(P,I_{i}, R_{i}) = \emptyset$.
\end{itemize}

The computation is said to converge to the set $IN_{\infty} = \bigcup_{i=0}^{\infty}IN_i$.

\end{definition}

\begin{example}
\label{ex:grand_exemple}

Let $P_{\ref{ex:grand_exemple}}$ be the following program:
 \begin{center}
      $\begin{Bmatrix}
	  \begin{array}{l}
		\outtextfact{n(1)} \\
		\outtextrule{n(X+1)}{n(X), (X+1)<=2} \\
		\outtextrule{a(X)}{n(X), not\ b(X), not\ b(X+1)} \\
		\outtextrule{b(X)}{n(X), not\ a(X)} \\
		\outtextrule{c(X)}{n(X), not\ b(X+1)}
	  \end{array} 
	\end{Bmatrix}$
 \end{center}


The following sequence is an \asperix. computation for $P_{\ref{ex:grand_exemple}}$:
\[\begin{array}{ll}
I_0 &= \partialinter{\emptyset}{\{\bot\}} \\
\\
r_1 &=  \outtextfact{n(1)} \in \Delta_{pro}(P_{\ref{ex:grand_exemple}},I_0, \emptyset)\\
I_1 &= \partialinter{\{n(1)\}}{\{\bot\}} \\
\\
r_2 &= \outtextrule{n(2)}{n(1)} \in \Delta_{pro}(P_{\ref{ex:grand_exemple}},I_1, \{r_1\})\\
I_2 &= \partialinter{\{n(1),n(2)\}}{\{\bot\}} \\
\\
& \Delta_{pro}(P_{\ref{ex:grand_exemple}},I_2, \{r_1,r_2\}) = \emptyset \\
r_3 &= \outtextrule{a(1)}{n(1), not\ b(1), not\ b(2)} \in \Delta_{cho}(P_{\ref{ex:grand_exemple}},I_2, \{r_1,r_2\})\\
I_3 &= \partialinter{\{n(1),n(2),a(1)\}}{\{\bot,b(1),b(2)\}} \\
\\
r_4 &= \outtextrule{c(1)}{n(1), not\ b(2)} \in \Delta_{pro}(P_{\ref{ex:grand_exemple}},I_3, \{r_1,r_2,r_3\})\\
I_4 &= \partialinter{\{n(1),n(2),a(1),c(1)\}}{\{\bot,b(1),b(2)\}} \\
\\
&\Delta_{pro}(P_{\ref{ex:grand_exemple}},I_4, \{r_1,r_2,r_3,r_4\}) = \emptyset \\
r_5 &= \outtextrule{a(2)}{n(2), not\ b(2), not\ b(3)} \in \Delta_{cho}(P_{\ref{ex:grand_exemple}},I_4, \{r_1,r_2,r_3,r_4\}) \\
I_5 &= \partialinter{\{n(1),n(2),a(1),c(1),a(2)\}}{\{\bot,b(1),b(2),b(3)\}}\\
\end{array}\]
\[\begin{array}{ll}
r_6 &= \outtextrule{c(2)}{n(2), not\ b(3)} \in \Delta_{pro}(P_{\ref{ex:grand_exemple}},I_5, \{r_1,r_2,r_3,r_4,r_5\})\\
I_6 &= \partialinter{\{n(1),n(2),a(1),c(1),a(2),c(2)\}}{\{\bot,b(1),b(2),b(3)\}}\\
\\
&\Delta_{pro}(P_{\ref{ex:grand_exemple}},I_6, \{r_1,r_2,r_3,r_4,r_5,r_6\}) = \emptyset \\
&\Delta_{cho}(P_{\ref{ex:grand_exemple}},I_6, \{r_1,r_2,r_3,r_4,r_5,r_6\}) = \emptyset \\
I_7 & =I_6\\
\end{array}\]

The previous \asperix. computation converges to the set \\ $\{n(1), n(2), a(1), c(1), a(2), c(2)\}$ which is an answer set for $P_{\ref{ex:grand_exemple}}$.

\end{example}

The following theorem establishes a connection between the results of any \asperix. computation and the answer sets of a normal logic program.

\begin{theorem}
\label{the:Aspcomp}

Let $P$ be a normal logic program and $X$ be an atom set.  
Then,  $X$ is an answer set of $P$ if and only if there is an \asperix. computation $\computation{R_i, I_i}$, $I_i=\partialinter{IN_i}{OUT_i}$, for $P$ such that $IN_{\infty} = X$.
\end{theorem}

Let us note that in order to respect the revision principle of an \asperix. computation each sequence of partial interpretations must be generated by using the propagation inference based on rules from $\Delta_{pro}$ as long as possible before using the choice based on $\Delta_{cho}$ in order to fire a nonmonotonic rule. 
Then, because of the non determinism of the selection of rules from $\Delta_{cho}$, the natural implementation of this approach leads to a usual search tree where, at each node, one has to decide whether or not to fire a rule chosen in $\Delta_{cho}$.
Persistence of applicability of the nonmonotonic rule chosen to be fired is ensured by adding to $OUT$ all ground atoms from its negative body. 
On the other branch, where the rule is not fired, the translation of its negative body into a new constraint ensures that it becomes impossible to find later an answer set in which this rule is not blocked.

Propagation can be improved by using ``must-be-true''\footnote{The term ``must be true'' is first used in \cite{DLVmbt}.} atoms:
atoms which have to be in the answer set to avoid a contradiction or, in other words,
atoms already determined to be in \IN. but which are not yet be proved to be in.

\begin{example}
\label{ex:must_be_true}
Let \intextrule{\bot}{not\ b} be a constraint whose body contains only one literal ${not\ b}$ with $b \not\in IN \cup OUT$.
In order to have an answer set, $b$ must be in $IN$ so that the constraint is not applicable
but $b$ is not yet proved (it is not the head of a fired rule).
Thus, one can only conclude that $b$ must be true.

\end{example}

Must-be-true atoms can be used during the propagation step in order to reduce the search space.

\begin{example}
\label{ex:must_be_truePropagation}
Let \intextrule{c}{a, b} be a rule with $a \in IN$ and 
$b \not\in IN$ but $b$ has been determined to be a must-be-true atom.
The rule may be fired during the propagation step 
but one can only conclude that the rule head $c$ must be true (because $b$ is not yet proved).

\end{example}

Must-be-true atoms can also be used to reduce the size of $\Delta_{cho}$, the set of nommonotonic rules that can be chosen to be fired.

\begin{example}
\label{ex:must_be_trueChoice}
Let \intextrule{c}{a, not~d} be a rule with $a \in IN$ and 
$d \not\in IN$ but $d$ has been determined to be a must-be-true atom.
The rule may already be considered to be blocked, even if $d$ is not yet proved, and thus may be excluded from $\Delta_{cho}$.

\end{example}

Note that must-be-true atoms are first used to improve propagation and choice but have to be proved later, otherwise the
computation can not lead to an answer set.

Notions of partial interpretation, rule status and \asperix. computation can be modified in order to consider these new elements.

\begin{definition}
\label{def:mbtInterpretation}
 Let $P$ be a logic program. A \emph{mbt partial interpretation} for $P$ is a triplet $\langle IN, MBT, OUT \rangle$ of disjoint atom sets included in the Herbrand base of $P$. 

\end{definition}

\begin{definition}
\label{def:mbtrulestatus}

  Let $r$ be a ground rule and $I = \langle IN, MBT, OUT\rangle$ be
  a mbt partial interpretation. 
  \begin{itemize}
  \item $r$ is \emph{supported} w.r.t. $I$ when \(\BODY.^+(r) \subseteq IN\),
  \item $r$ is \emph{weakly supported} w.r.t. $I$ when \(\BODY.^+(r)  \subseteq (IN \cup MBT)\)
  \item $r$ is \emph{blocked} w.r.t. $I$ when \(\BODY.^-(r) \cap (IN \cup MBT) \neq \emptyset\),
  \item $r$ is \emph{unblocked} w.r.t. $I$ when \(\BODY.^-(r) \subseteq OUT\),
  \item $r$ is \emph{applicable} w.r.t. $I$ when $r$ is supported and not blocked.
  \end{itemize}
\end{definition}

Propagation is extended by Mbt-propagation: if some rule is weakly supported and unblocked w.r.t. mbt partial interpretation $\langle IN, MBT, OUT \rangle$ (but is not supported, i.e., does not belong to $\Delta_{pro}$), then the head of the rule can be added in $MBT$ set.
And $\Delta_{cho}$, the set of rules that can be chosen, is restricted to the rules that are not blocked w.r.t. mbt partial interpretation.

\begin{definition}
\label{def:Delta_cho_pro_mbt}
Let $P$ be a set of first order rules, $I = \langle IN, MBT, OUT \rangle$  be a mbt partial interpretation and $R$ be a set of ground rules.
\begin{itemize}
\item $\Delta_{pro}(P,I, R) =\{r \in  ground(P) \setminus R \mid \BODY.^+(r) \subseteq IN \mbox{ and } \BODY.^-(r) \subseteq OUT\} $
\item $\Delta_{pro\_mbt}(P,I, R) =\{r \in  ground(P) \setminus R \mid \BODY.^+(r) \subseteq IN \cup MBT, \BODY.^+(r) \not\subseteq IN \mbox{ and } \BODY.^-(r) \subseteq OUT\} $
\item $\Delta_{cho\_mbt}(P,I, R) =\{r \in  ground(P) \setminus R \mid \BODY.^+(r) \subseteq IN,  \mbox{ and } \BODY.^-(r) \cap (IN \cup MBT) = \emptyset\} $
\end{itemize}

\end{definition}

A mbt \asperix. computation is an \asperix. computation with this additional kind of propagation and with the possibility to block a rule from $\Delta_{cho\_mbt}$ instead of firing it (``Rule exclusion"). 
To block a rule is to add a constraint with the negative literals of the rule body. 
If there is only one literal in the negative body, this constraint can be expressed by adding an atom in MBT set (see Example \ref{ex:must_be_true}).
These possibilities restrict rule choice in $\Delta_{cho\_mbt}$ and thus forbid some computations: if a rule $r$ is blocked, computation can only converge to an answer set whose generating rules do not contain $r$.
Note that Convergence principle impose that, at the end of a computation, no constraint is applicable
and each atom from MBT set has been proved (i.e., was moved from MBT to IN set).

\begin{definition}
\label{def:mbtComputation}
 Let $P$ be a first order normal logic program.
  A \emph{mbt \asperix. computation} for $P$ is a sequence ${\langle K_i, R_i , I_i \rangle}_{i=0}^{\infty}$
  of ground rule sets $K_i$ and $R_i$ and mbt partial interpretations $I_i=\langle IN_i, MBT_i, OUT_i \rangle$  that satisfies the following conditions:
  \begin{itemize}
  \item $K_0 = \emptyset$, $R_0 = \emptyset$ and $I_0 = \langle \emptyset, \emptyset, \{\bot\} \rangle$,
  \item (Revision) $\forall i \geq 1$, 
	\begin{itemize}
	\item[](Propagation) $K_i = K_{i-1}$,\\
	$R_i = R_{i-1} \cup \{r_i\}$ with $r_i \in \Delta_{pro}(P,I_{i-1}, R_{i-1})$\\
	and  $I_i = \langle IN_{i-1} \cup \{\HEAD.(r_i)\}, MBT_{i-1}  \setminus \{\HEAD.(r_i)\} , OUT_{i-1} \rangle$
	\item[or] (Mbt-propagation) $K_i = K_{i-1}$, $R_i = R_{i-1}$,\\
	and  $I_i = \langle IN_{i-1}, MBT_{i-1} \cup \{\HEAD.(r_i)\}, OUT_{i-1} \rangle$\\
	with $r_i \in \Delta_{pro\_mbt}(P,I_{i-1}, R_{i-1})$
	\item[or](Rule choice) $\Delta_{pro}(P \cup K_{i-1},I_{i-1}, R_{i-1}) = \emptyset$, \\
	$\Delta_{pro\_mbt}(P\cup K_{i-1},I_{i-1}, R_{i-1}) = \emptyset$,\\
	$K_i = K_{i-1}$,\\
	$R_i = R_{i-1} \cup \{r_i\}$ with $r_i \in \Delta_{cho\_mbt}(P,I_{i-1}, R_{i-1})$\\
	and $I_i = \langle IN_{i-1} \cup \{\HEAD.(r_i)\}, MBT_{i-1}  \setminus \{\HEAD.(r_i)\} ,OUT_{i-1}\cup \BODY.^-(r_i) \rangle$
	\item[or]  (Rule exclusion) $\Delta_{pro}(P\cup K_{i-1},I_{i-1}, R_{i-1}) = \emptyset$,\\
	$\Delta_{pro\_mbt}(P\cup K_{i-1},I_{i-1}, R_{i-1}) = \emptyset$,\\
	\begin{itemize}
		\item[]
		$K_i = K_{i-1}$, $R_i = R_{i-1}$\\
		and $I_i = \langle IN_{i-1}, MBT_{i-1}  \cup \BODY.^-(r_i) ,OUT_{i-1} \rangle$ \\
		with $r_i \in \Delta_{cho\_mbt}(P,I_{i-1}, R_{i-1})$ and $|\BODY.^-(r_i)| = 1$\\
		\item[or]
		$K_i = K_{i-1} \cup \{\bot \leftarrow {\cup}_{b \in \BODY.^-(r_i)} not~b.\}$, $R_i = R_{i-1}$ and $I_i = I_{i-1}$ \\
		with $r_i \in \Delta_{cho\_mbt}(P,I_{i-1}, R_{i-1})$ and $|\BODY.^-(r_i)| > 1$
	\end{itemize}
	\item[or] (Stability) 
	$K_i = K_{i-1}$, $R_i = R_{i-1} $ and $I_i = I_{i-1}$,
	\end{itemize}
  \item (Convergence) $\exists i \geq 0, ~\Delta_{cho\_mbt}(P\cup K_i,I_{i}, R_{i}) = \emptyset \mbox{ and } MBT_i = \emptyset$.
\end{itemize}

\end{definition}

Mbt \asperix. computations characterize answer sets of a normal logic program.
Completeness and correctness are established by the following theorem.

\begin{theorem}
\label{the:mbtAsperixcomp}

Let $P$ be a normal logic program and $X$ be an atom set.  
Then,  $X$ is an answer set of $P$ if and only if there is a mbt \asperix. computation ${\langle K_i, R_i, I_i \rangle}_{i=0}^{\infty}$, $I_i = \langle IN_i, MBT_i, OUT_i\rangle$, for $P$ such that $IN_{\infty} = X$.
\end{theorem}

Note that computations model only successfull branches of a search tree.
On the other hand, must-be-true atoms and rules blocking enable to prune failed branches of the tree
and to reduce non determinism of the search by restricting the possible choices for the oracle 
(because some rules are explicitly excluded, and others are blocked by must-be-true atoms).
So, these new elements do not improve the number of steps of a computation
but they improve the number of steps needed to find a computation when there is no oracle to guide the search
and, then, they make easier the search of answer sets.

\subsection{\asperix. Main Algorithm}
\label{subsec:core}
Now, we are interested in the practical computation of an answer set. The \asperix. algorithm, following the principle of mbt \asperix. computation seen
in section \ref{subsec:computation}, is based on the construction of three disjoint atom sets \IN., \MBT. and \OUT. during the search for an answer set. It alternates two steps.
On the one hand, a propagation step which instantiates all supported and unblocked rules which may be built from \IN., \MBT. and \OUT. and fires them, i.e. adds their head
in \IN. (or \MBT.).
On the other hand, a choice step which forces or prohibits a nonmonotonic instantiated applicable rule to be fired during the next propagation step.

In order to treat the information more efficiently, the rules of a program $P$ are ordered following the \emph{strongly connected components} ($SCC$) of the
dependency graph of $P$:
the nodes of the dependency graph of a program $P$ are its predicate symbols and the arcs are defined by $\{(p,q) | \exists r \in P,\ p = pred(\head{r}),\ q \in
pred(\corpspos{r}\cup\corpsneg{r})\}$.
The strongly connected components $\{C_1, ..., C_n\}$ are ordered in such a way that if $i<j$ then no node (i.e. predicate symbol) of $C_i$ depends of a node of $C_j$.
A rule is said to belong to a SCC $C$ if the predicate symbol of its head is in the component $C$.
Note that constraints are not really concerned by ordering of rules but, for standardizing notations,
constraints are considered to belong to a unique component whose number is greater than that of the last SCC,
i.e., if $C_n$ is the last SCC then constraints are considered to belong to $C_{n+1}$.

\begin{example}
\label{ex:SCC}
(Example~\ref{ex:grand_exemple} continued)

The strongly connected components (SCC) of the graph of the program $P_{\ref{ex:grand_exemple}}$ are $C_1 = \{n\}$, $C_2 = \{a,b\}$ and $C_3 = \{c\}$ (Figure \ref{graphe_dependance}).

\begin{figure}[!ht]
      \centering
      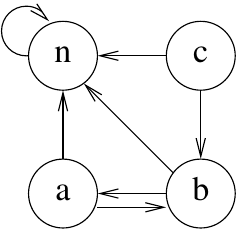
      \caption{Dependency graph of $P_{\ref{ex:grand_exemple}}$} 
      \label{graphe_dependance} 
\end{figure}

\end{example}

The \asperix. algorithm solves one by one the SCC $\{C_1, ..., C_n\}$ of a program $P$ by starting by $C_1$.
When no propagation nor choice can no longer be done on the current SCC, the predicate symbols of the SCC are said to be \emph{solved} and the SCC too.
It means that nothing can be deduced anymore for those predicate symbols.
The atoms which are instances of the predicate symbols of the current SCC and which are not in \IN. are implicitly added to \OUT..
Note that they are not explicitly added to \OUT. because ground instances of a predicate are not known (not computed): they could be
infinite and, even if finite, to compute and store them is useless.

Rules of the program are instantiated on the fly during the propagation phase and the choice step.
Hence the propositional program $ground(P)$  which contains all the instantiated rules of the program is never really computed.
The propagation step and the choice step are realized in the \asperix. algorithm thanks to the functions \gammapro. and \gammacho. (which are selection functions in $\Delta_{pro} \cup \Delta_{pro\_mbt}$ and $\Delta_{cho\_mbt}$ sets used in the mbt \asperix. computation of Subsection~\ref{subsec:computation}).
The  \gammapro. function searches for a weakly supported unblocked rule amongst the current and next non-solved SCC.
So propagation operates on several components: each rule is fired as soon as possible to quickly detect a possible conflict.
Rules instantiated by \gammapro. are stored in a set $Subst=\{subst(r_i), ...,subst(r_n)\}$ during the answer set search to mark the substitution of rules that have
already been used.
For each first-order rule $r_i$, $subst(r_i)$ denotes the set of all substitutions $\theta$ such that $\theta(r_i)$ has already been fired. And $subst\_rule(r_i) = \bigcup_{\theta \in subst(r_i)} \{\theta(r_i)\}$ is the set of instantiated rules obtained thanks to substitutions $subst(r_i)$. 
The notation is extended to a set $R$ of first-order rules: $subst\_rule(R) = \bigcup_{r_i \in R}subst\_rule(r_i)$.
The \gammacho. function chooses an applicable rule in the current SCC when nothing can no longer be propagated.
So choice, unlike propagation, operates only on the current component. 
This strategy, consisting of solving the SCC one after another, makes it possible to solve efficiently stratified programs (or some stratified parts of programs). 

Functions \gammapro. and \gammacho. are specified in more details in Subsection~\ref{subsec:gamma} and are defined informally as follows:

\begin{itemize}
 \item $\gammapro.(P,S,S',T,SCC,Subst)$: nondeterministic function which selects a rule (or a constraint) $r$ belonging to a SCC greater or equal to the current SCC in the dependency graph of a program $P$ such that $\corpspos{r} \subseteq S\cup S'$, $\corpsneg{r} \subseteq T$ and $r \in  ground(P) \backslash subst\_rule(P)$ or returns $\NULL.$ if no such a rule exists.
 \item $\gammacho.(P,S,S',T,SCC,Subst)$: nondeterministic function which selects a rule $r$ belonging to the current SCC in the dependency graph of a program $P$ such that  $\corpspos{r} \subseteq S$, $\corpsneg{r} \cap (S \cup S') = \emptyset$ and $r \in  ground(P) \backslash subst\_rule(P)$ or returns $\NULL.$ if no such a rule exists.
\end{itemize}

\begin{algorithm}[!ht]
\label{algo:solve}
\LinesNumbered
  \FuncSty{Function solve($P_R,~ P_K,~ IN,~ MBT,~ OUT, ~ SCC, ~ Subst$)}\;
// search of one answer set for a program $P = P_R \cup P_K$\\
  \Repeat(//Propagation phase){$r_0 = \NULL.$}{ 
    $r_0 \leftarrow \gamma_{pro}(P_R \cup P_K,IN, MBT,OUT,SCC,Subst)$\; \label{ln:gamma-pro}
    \If {$r_0 \neq \NULL.$}{
      \eIf {$(\corpspos{r_0}\cap \MBT.)\neq \emptyset$} 
      {
	$MBT \leftarrow MBT \cup \{\head{r_0}\}$\;\label{ln:add-mbt}
      } 
      {
	$IN \leftarrow IN \cup \{\head{r_0}\}$\;\label{ln:add-in}
	\If {$(\head{r_0} \in MBT$)} 
	{\label{ln:del-mbt}
	  $MBT \leftarrow MBT \backslash \{\head{r_0}\}$\;
	}
      }
    }
  }
  \eIf (//Contradiction detected){$((IN \cup MBT) \cap OUT \neq \emptyset)$}{\label{ln:contradiction}
    \Return $no\_answer\_set$\; } 
  {
    $r_0 \leftarrow \gamma_{cho}(P_R,IN,MBT,OUT,SCC,Subst)$\;\label{ln:gamma-cho}
    \eIf (//Choice point) {$r_0 \neq \NULL.$} 
     {
      $stop \leftarrow solve(P_R, P_K, IN, MBT, OUT \cup \corpsneg{r_0}, SCC,Subst)$\;\label{ln:force-r0}
      \If {$stop = no\_answer\_set$} 
      {\label{ln:block-r0begin}
	$atoms \leftarrow \{a | a \in \corpsneg{r_0}, pred(a) \in pred(SCC)\}$\;
        \eIf {$(| atoms | = 1)$}
        {
	  $MBT \leftarrow MBT \cup atoms$\; \label{ln:block-atom}
        }
        {
	  $P_K \leftarrow P_K \cup \{ \bot \leftarrow \cup_{a_i \in atoms}\ not\ a_i\}$\;\label{ln:block-atoms}
        }
        $stop \leftarrow solve(P_R, P_K, IN, MBT, OUT, SCC,Subst)$\;\label{ln:block-r0end}
      }
      \Return $stop$ \;
    }
    (// The SCC is solved) {\label{ln:solved-scc}
      \eIf {$pred(\MBT.)\cap pred(SCC) = \emptyset$} 
      {\label{ln:mbt-check}
        \eIf{$\neg last(SCC)$}
        {
	  \Return $solve(P_R, P_K, IN, MBT, OUT, SCC+1,Subst)$\;
        }
        {
	  \eIf (// a constraint is violated){$\gamma_{check}(P_K,IN,MBT,OUT,SCC)$}
	  {\label{ln:gamma-check}
	    \Return $no\_answer\_set$\; 
	  }
	  (// An answer set has been found){
	    \Return $IN$\; 
	  }
        }
      }
      (// a \MBT. atom can not be proved){\label{ln:mbt-fail}
        \Return $no\_answer\_set$\; 
      }
    }

  }
\caption{$solve$}
\end{algorithm}

The function $solve$ of Algorithm~\ref{algo:solve} specifies the algorithm of the search of one answer set for a program $P$.
The set $P_K$ is the set of constraints (rules with the symbol $\bot$ at their heads) of $P$ and $P_R$ the other rules.
By default, $\bot$ is into the set \OUT..
Then, if a constraint is fired, a contradiction is immediately detected since $\bot$ is added into the set \IN. and the sets \IN. and \OUT. are no longer disjoint.
The algorithm of the function \solve. computes one answer set (or none if the program is incoherent) thanks to the variable $stop$ which stops the search once an answer set has been found.
This algorithm may be easily extended to compute an arbitrary number of answer sets.
Let us note that, for sake of simplicity, the function $solve$ will return either a set (when there is an answer set) or the constant $no\_answer\_set$ if there is no answer set. 

The main parts of the function \solve. are now described.
Initially, $IN = \emptyset$, $MBT = \emptyset$, $OUT = \{\bot\}$, $SCC$ is the index of the first SCC and $Subst = \emptyset$.

The propagation phase successively fires each weakly supported and unblocked instantiated rule $r_0$.
At each step, the call $\gammapro.(P_R \cup P_K, IN , MBT, OUT, SCC, Subst)$  selects and instantiates a unique unblocked rule $r_0$ such that $\corpspos{r_0} \subseteq \IN. \cup \MBT.$ (line~\ref{ln:gamma-pro}).
If such a rule exists, its head atom \head{r_0} must belong to the answer set.
This head atom is added into the set \IN. (line~\ref{ln:add-in}) if the positive body of the rule is included in the set \IN. or added into the set \MBT. (line~\ref{ln:add-mbt}) otherwise since at least one atom $a$ of the positive body of the rule has not yet proved its membership to the set \IN. ($a\in \MBT.$ but $a\notin \IN.$). 
Moreover, a head atom which is added into \IN. must be deleted from \MBT. since a proof of its membership to the answer set has been found (line~\ref{ln:del-mbt}).
When there is no more unblocked rule $r_0$ such that $\corpspos{r_0} \subseteq \IN. \cup \MBT.$, $(\IN. \cup \MBT.)\cap \OUT. = \emptyset$ is checked in order to detect a contradiction (line~\ref{ln:contradiction}).
If no contradiction is detected, the algorithm begins the choice step.

The choice point forces or forbids a nonmonotonic applicable rule to be fired.
The call $\gammacho.(P_R, ~IN, ~MBT,~OUT, ~SCC, ~Subst)$ selects and instantiates a unique applicable rule of $P_R$ whose head belongs to the current SCC (line~\ref{ln:gamma-cho}). If such a rule exists, $r_0$ is forced to be unblocked and then will be fired during the next propagation phase: 
its negative body is added to the \OUT. set and function \solve. is recursively called with its new parameters (line~\ref{ln:force-r0}).
If a recursive call to the function \solve. detects a contradiction, the algorithm backtracks on the last choice point on the  rule $r_0$ which has been forced to be fired and blocks it (lines~\ref{ln:block-r0begin}-\ref{ln:block-r0end}): if $a$ is the only atom of the negative body of $r_0$  then $a$ is added to the set \MBT. (line~\ref{ln:block-atom}) else a constraint including all the atoms of the negative body of $r_0$ is added to the program (line~\ref{ln:block-atoms}).
More precisely, the only atoms of the negative body that are considered are those with a predicate symbol belonging to the current SCC because atoms from a lower SCC are already solved, i.e. they are in \IN. or \OUT.. 
When there is no more choice point, the current SCC is solved (line~\ref{ln:solved-scc}) but it must be checked that no atom of the \MBT. set has a predicate symbol in the current SCC (line~\ref{ln:mbt-check}).
If such an atom exists, 
\MBT. and \OUT. sets are not disjoint.
Indeed, if a SCC is solved, atoms which are instances of predicate symbols of the SCC and which are not in \IN. are implicitly added to \OUT..
Then if a \MBT. atom is an instance of a predicate symbol of the current SCC, a failure is observed and the backtrack process continues (line~\ref{ln:mbt-fail}).
If the last SCC is solved, the set \IN. represents an answer set of $P$ if no constraint is applicable.
This test is realized thanks to the nondeterministic function \gammacheck. (line~\ref{ln:gamma-check}) which is specified in more details in Subsection~\ref{subsec:gamma} and is defined informally as follows:

$\gammacheck.(P,S,S',T, SCC)$: function which checks if there is any constraint $c$ such that $\corpspos{c} \subseteq S$, $\corpsneg{c} \cap S = \emptyset$ and $c \in  ground(P)$.

\begin{example} 
\label{ex:trace_grand_exemple}

\begin{figure}[b]
	\hspace*{-4cm}
      \includegraphics[width=2.2\textwidth]{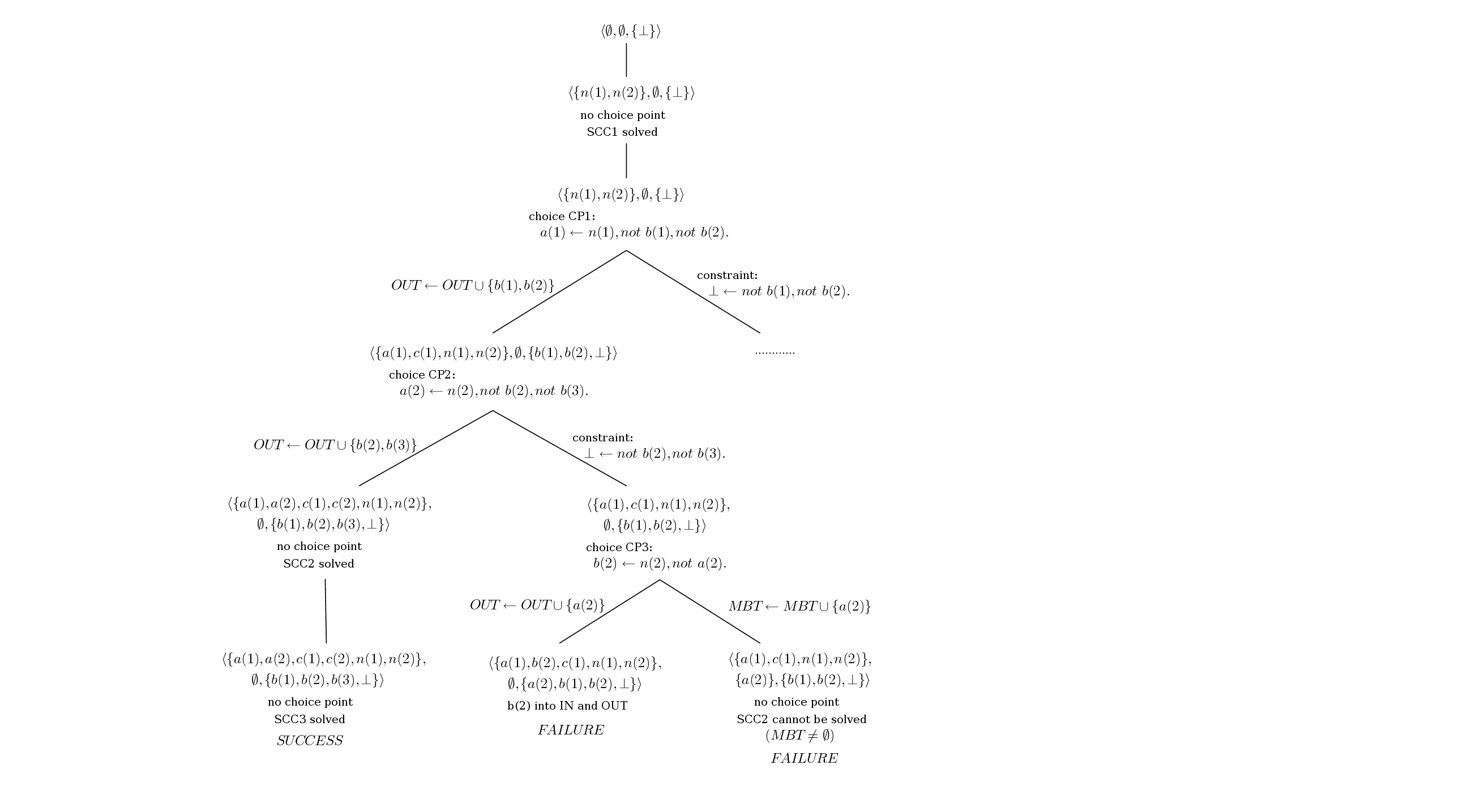}
            \caption{The tree-shaped execution of the answer sets of program $P_{\ref{ex:grand_exemple}}$.}
      \label{fig:trace} 
\end{figure}   

The execution of the \asperix. algorithm for program $P_{\ref{ex:grand_exemple}}$ of Example~\ref{ex:grand_exemple} is represented by a tree in Figure~\ref{fig:trace}.
At the beginning $\IN. = \emptyset$, $\MBT. = \emptyset$, $\OUT. = \{ \bot \}$ and the current SCC is the component $C_1 = \{n\}$.
After the first propagation, $n(1)$ and $n(2)$ are in \IN. thanks to the two rules \intextfact{n(1)} and \intextrule{n(2)}{n(1),(1+1) <= 2}. 
No choice point exists and the first SCC is solved since the \MBT. set is empty.
The component $C_2=\{a,b\}$ becomes the current SCC.

The first choice is realized on the current SCC (choice point $CP1$): the rule  \intextrule{a(1)}{n(1), not \ b(1), not \ b(2)} becomes unblocked by adding $b(1)$ and $b(2)$ into the set \OUT. (left branch after choice point $CP1$).
A new propagation phase shows that $a(1)$ and $c(1)$ are in \IN. since \intextrule{a(1)}{n(1), not \ b(1), not \ b(2)} and \intextrule{c(1)}{n(1), not \ b(2)} can be fired. 
Then, a new choice is realized (choice point $CP2$) and the rule \intextrule{a(2)}{n(2), not \ b(2), not \ b(3)} is forced to be unblocked (left branch after choice point $CP2$). 
The atom $b(3)$ is added into the set \OUT..
A new propagation phase shows that $a(2)$ and $c(2)$ are in \IN. since \intextrule{a(2)}{n(2), not \ b(2), not \ b(3)} and \intextrule{c(2)}{n(2), not \ b(3)} can be fired.
The second SCC is solved since no other rule is applicable and the \MBT. set is still empty.
In the same way, no propagation nor choice point is possible in the SCC $C_3=\{c\}$.
Since no constraint is applicable, a first answer set is obtained: $\{ a(1),a(2),c(1),c(2),n(1),n(2) \}$. 

If another answer set is wished, the algorithm backtracks to the last choice point on \intextrule{a(2)}{n(2), not \ b(2), not \ b(3)} of the component $C_2$ and blocks the rule (right branch after choice point $CP2$) by adding a constraint \intextrule{\bot}{not \ b(2), not \ b(3)} into $P_K$.
A new choice is realized (choice point $CP3$) and the rule \intextrule{b(2)}{n(2), not \ a(2)} is forced to be unblocked (left branch after choice point $CP3$) by adding $a(2)$ into the \OUT. set.
During the propagation step, $b(2)$ is added into the \IN. set since \intextrule{b(2)}{n(2), not \ a(2)} is fired. 
The atom $b(2)$ is then simultaneously in the sets \IN. and \OUT. which leads to a contradiction.

The algorithm backtracks to the choice point \intextrule{b(2)}{n(2), not \ a(2)} of the component $C_2$ (choice point $CP3$) and the rule is blocked by adding the atom $a(2)$ into the \MBT. set (right branch after choice point $CP3$).
Since there is no more possible choice and the \MBT. set contains an atom whose predicate symbol is in the current SCC, this atom cannot be proved and this
leads to a failure.
The algorithm backtracks to the first choice point on \intextrule{a(1)}{n(1), not \ b(1), not \ b(2)} of the component  $C_2$ (choice point $CP1$) and blocks the rule and searches for a new possible answer set (right branch after choice point $CP1$). The process keeps going until the whole tree is computed when all the answer sets are required. Let us note that when dealing with the computation
of one answer set like explained in the algorithm, only the first branch is considered.

\end{example}

\subsection{Functions $\gamma$}
\label{subsec:gamma}

Functions $\gamma$ have a crucial role in two important steps of the search of an answer set.
The function \gammapro. is called during the propagation step in order to choose the rules to fire and then to add their heads into \IN. (or \MBT.).
The function \gammacho. is called during the choice step in order to force or to forbid a rule to be fired during the next propagation step.
The function \gammacheck. is called during the verification step in order to verify that no constraint is applicable.
 Since the principle of the solver \asperix. is to instantiate the rules on the fly during the search of an answer set, functions $\gamma$ need to call a function \instantiateRule. which searches for a substitution for the atoms of a rule.
This function is detailed in its own Subsection~\ref{subsec:instantiation}.

\paragraph{\bf {Function \gammapro..}}

The function \gammapro. searches for a rule to fire w.r.t. \IN., \MBT. and \OUT. sets.
This function computes a complete instantiation of a rule such that the positive body is in $IN \cup MBT$ and the negative body is in \OUT..
The rule to instantiate is chosen amongst a set of rules $R$ consisting of rules that could lead to new, unprocessed instances.
These rules are those whose body contains some predicate symbol of what we call an \emph{atom to propagate}.
Atoms to propagate are atoms recently added into \IN., \MBT. and \OUT. sets, and not yet used for propagation phase.
Thereby, when an atom $a$ is added into \IN. or \MBT. (resp. \OUT.) set, the rules containing $pred(a)$ in their positive body (resp. negative body) will be in the $R$ set for the next call to \gammapro. in order to propagate this atom, i.e. to use its presence in \IN. or \MBT. (resp. \OUT.) for building new instances of rules to be fired.
During the first call of the function \solve., atoms to propagate are the facts of the program, and the set $R$ contains all the rules which have some predicate symbols of the facts in their positive body.
During a call after a choice point, atoms to propagate are those added into \OUT. during this choice point, and the set $R$ contains all the rules which have  some predicate symbols of these atoms in their negative body.
During a call after the access to the next SCC, the  predicate symbols of the current SCC are solved and then 
all instances of these predicate symbols that are not in \IN. are implicitly added into \OUT..
Atoms to propagate are all these instances determined to be false, and then the set $R$ contains all the rules which have in their negative body some of these solved predicate symbols.

The Algorithm~\ref{algo:gamma_pro} of the function \gammapro. chooses a first-order rule $r$ amongst the set $R$ (the first one, line~\ref{ln:firstR}) and tries to find a weakly supported unblocked instantiation of  the rule.
It calls the function \instantiateRule. which returns this next instantiation if any (line~\ref{ln:instantiate}).
If there is no more weakly supported unblocked instantiated rule which may be extracted from $r$ (line~\ref{ln:no-instantiation}), \gammapro. deletes from $R$ the rule $r$ and treats the next rule.
This process is repeated until a weakly supported unblocked rule is found or there is no more rule in $R$.
When a rule allows a substitution (line~\ref{ln:instantiation-found}), the latter is stored in $Subst$ in order to find some others at the next call to \gammapro.. 

\begin{algorithm}[ht]
\label{algo:gamma_pro}
  
\FuncSty{Function $\gamma_{pro}(P,IN,MBT,OUT,SCC,Subst)$}\;
  $R \leftarrow$ Set of rules (including constraints) containing predicate symbols to propagate\;
  \eIf {$R \neq \emptyset$}
  {
      \Repeat{$\theta \neq \NULL.$ or $R = \emptyset$}{
	  $r \leftarrow first(R)$\; \label{ln:firstR}
	  \tcc{Searching for an instantiation of the rule $r$ with $\corpspos{r} \subseteq \IN. \cup \MBT.$ and $\corpsneg{r} \subseteq \OUT.$}
	  $\theta \leftarrow instantiateRule(r,\gamma_{pro},\IN.,\MBT.,\OUT.,subst(r))$\;\label{ln:instantiate}
	  \If{$\theta = \NULL.$}{\label{ln:no-instantiation}
	      $R \leftarrow R \backslash \{ r \}$\;
	  }
      }
      \eIf{$\theta \neq \NULL.$}
      {\label{ln:instantiation-found}
	  \tcc{An unblocked weakly supported instantiated rule is found}
	  $subst(r) \leftarrow subst(r) \cup \{ \theta \}$\;
	  \Return $\theta(r)$\;
      }
      {
	  \Return $\NULL.$\;
      }
  }
  {
    \Return $\NULL.$\;
  }
\caption{$\gamma_{pro}$}  

\end{algorithm}

\begin{example}
\label{ex:gamma_pro}
Example~\ref{ex:trace_grand_exemple} is taken again. An answer set is searched after the choice point on the rule \intextrule{a(1)}{n(1), not\ b(1), not\ b(2)}
(choice point $CP1$): the atoms $b(1)$ and $b(2)$ are added into the set \OUT. in order to force the rule to be fired (left branch after choice point $CP1$).
During the propagation step, many calls to the function \gammapro. are executed.
During the first call the set $R$ consists of all the rules containing in their negative body the predicate symbol $b$ of the atoms $b(1)$ and $b(2)$ that must be propagated.
This set $R$ then contains the rules \intextrule{a(X)}{n(X), not\ b(X), not\ b(X+1)} and \intextrule{c(X)}{n(X), not\ b(X+1)}
Arbitrarily, the rule \intextrule{a(X)}{n(X), not\ b(X), not\ b(X+1)} of the set $R$ is chosen and a supported unblocked instantiation \intextrule{a(1)}{n(1), not\ b(1), not\ b(2)} is found.
The function \gammapro. returns the instantiation of the rule and the \solve. function adds  $a(1)$ into \IN..
During the next call to \gammapro., the set $R$ must contain, in addition to the previous rules, any rule containing in its positive body the predicate symbol $a$ of the atom to be propagated $a(1)$ (since $a(1)$ has been added into \IN.).
Since no rule respects this condition, the set $R$ still contains only the two previously added rules.
The function \gammapro. searches for a new weakly supported unblocked instantiation of the rule  \intextrule{a(X)}{n(X), not\ b(X), not\ b(X+1)}.
No such instantiation is found and the rule is deleted from the set $R$.
The function \gammapro. searches for a new weakly supported unblocked instantiation of the rule \intextrule{c(X)}{n(X), not\ b(X+1)}.
The instantiation \intextrule{c(1)}{n(1), not\ b(2)} is then returned to the \solve. function which adds $c(1)$ into \IN..
Then during the next call to the function \gammapro., the set $R$ must be updated with the rules containing in their positive body the predicate symbol $c$ of the atom to propagate $c(1)$.
As previously, no rule respects this condition and the set $R$ still contains the only rule \intextrule{c(X)}{n(X), not\ b(X+1)}.
A new weakly supported unblocked instantiation is sought but this rule leads to a failure. 
The rule  \intextrule{c(X)}{n(X), not\ b(X+1)} is then deleted from the set $R$ which becomes empty.
Then the function \gammapro. returns the value \NULL. and the propagation step of the function \solve. stops.

 \end{example}

\paragraph{\bf Function \gammacho..}

The function \gammacho. is executed when no rule can be fired anymore and there is some SCC to be solved.
This function searches an applicable instantiated rule belonging to the current SCC.
The Algorithm~\ref{algo:gamma_cho} of function \gammacho. is similar to the algorithm of the function \gammapro..
The function \gammacho. searches for an applicable instantiated rule amongst a set $R$ of rules which have in their negative body at least one predicate symbol from the current SCC 
(otherwise, if all predicate symbols from negative body belong to previous SCC, they are already solved and then the rule can be considered as a monotonic one and is only used for propagation).
The function \gammacho. chooses a rule in this set $R$ before calling the function \instantiateRule. searching for the next applicable instantiation for the considered rule.
In a similar way as the function \gammapro., the process is repeated until an applicable instantiated rule is found for a rule of $R$ or there is no more rule in $R$.

\begin{algorithm}[ht]
\label{algo:gamma_cho}
  \FuncSty{Function $\gamma_{cho}(P,\IN.,\MBT.,\OUT.,SCC,Subst)$}\;
  $R \leftarrow$ Set of rule belonging to the current SCC such that the negative body contains at least a predicate symbol not solved\;
  \eIf {$R \neq \emptyset$}
  {
     \Repeat{$\theta \neq \NULL.$ or $R = \emptyset$}{
	  $r \leftarrow first(R)$\;
	  \tcc{Searching for an instantiation of the rule $r$  with $\corpspos{r} \subseteq \IN.$ and $\corpsneg{r} \cap (\IN.\cup \MBT.) = \emptyset$}
	  $\theta \leftarrow instantiateRule(r,\gamma_{cho},\IN.,\MBT.,\OUT., subst(r))$\;
	  \If{$\theta = \NULL.$}{
	      $R \leftarrow R \backslash \{ r \}$\;
	  }
     }
     \eIf{$\theta \neq \NULL.$}
      {
	  \tcc{An applicable instantiated rule is found}
	  $subst(r) \leftarrow subst(r) \cup \{ \theta \}$\;
	  \Return $\theta(r)$\;
      }
      {
	  \Return $\NULL.$\;
      }
  }
  {
    \Return $\NULL.$\;
  }
  
\caption{$\gamma_{cho}$}  

\end{algorithm}

\begin{example}
\label{ex:gamma_cho}
Example~\ref{ex:gamma_pro} is taken again.
After the first SCC has been solved, a first choice is realized on the current SCC,  $C_2 = \{a,b\}$, by the function \gammacho..
The rules of this component which contains in their negative body at least one predicate symbol $a$ or $b$ of $C_2$ are added into the set $R$ of the rules that may be chosen.
Then, the rules \intextrule{a(X)}{n(X), not\ b(X), not\ b(X+1)} and \intextrule{b(X)}{n(X), not\ a(X)} are in $R$.
Arbitrarily, the function \gammacho. searches for an applicable instantiation of the first  rule of this set and a choice point on \intextrule{a(1)}{n(1), not\ b(1), not\ b(2)} is returned to the calling function \solve. (choice point $CP1$).
After the propagation step, \gammacho. searches for a new applicable instantiation of the rule \intextrule{a(X)}{n(X), not\ b(X), not\ b(X+1)} and a choice point on \intextrule{a(2)}{n(2), not\ b(2), not\ b(3)} is returned to the calling function \solve. (choice point $CP2$).
After a new propagation step, \gammacho. searches in vain a new applicable instantiation of the rule \intextrule{a(X)}{n(X), not\ b(X), not\ b(X+1)} 
This last rule is then deleted from the set $R$ and \gammacho. searches for an applicable instantiation of the rule \intextrule{b(X)}{n(X), not\ a(X)} which leads to a failure.
The set $R$ is now empty and the function \gammacho. returns \NULL. to the calling function \solve. to mean that no other choice may be realized on the current SCC.

\end{example}

\paragraph{\bf Function \gammacheck..}

The function \gammacheck. is executed when no more choice point is possible for the last SCC.
This function verifies that no constraint containing at least one predicate symbol of the last SCC is applicable in order to determine if the set \IN. is an answer set.
The Algorithm~\ref{algo:gamma_check} of the function \gammacheck. is similar to the algorithm of the function \gammacho.. 
The function \gammacheck. searches for an applicable instantiated constraint amongst a set $C$ of constraints whose negative body contains at least a not-solved predicate symbol of the last SCC.
The function \gammacheck. chooses a constraint in the set $C$ and calls the function \instantiateRule. which searches for an applicable instantiated constraint.
If no instantiated constraint is applicable, the algorithm returns \false. and the set \IN. is an answer set of the program. If a constraint is applicable, the algorithm
returns \true. which means there is a failure on the branch (the search of answer sets keeps going on other branches if any).

\begin{algorithm}[ht]
\label{algo:gamma_check}

  \FuncSty{Function $\gamma_{check}(P,IN,MBT,OUT,SCC)$}\;
  $C \leftarrow$ Set of constraints such that the negative body contains at least a predicate symbol not solved\;
  \eIf {$C \neq \emptyset$}
  {
     \Repeat{$\theta \neq \NULL.$ or $C = \emptyset$}{
	  $c \leftarrow first(C)$\;
	  \tcc{Searching for an instantiation of the constraint $c$ such that $\corpspos{c} \subseteq \IN.$ and $\corpsneg{c} \cap \IN. = \emptyset$}
	  $\theta \leftarrow instantiateRule(c,\gamma_{check},\IN.,\MBT.,\OUT.,\emptyset)$\;
	  \If{$\theta = \NULL.$}{
	      $C \leftarrow C \backslash \{ c \}$\;
	  }
     }
     \eIf{$\theta \neq \NULL.$}
      {
	  \tcc{An applicable instantiated constraint is found}
	  \Return \true.\;
      }
      {
	  \Return \false.\;
      }
  }
  {
    \Return \false.\;
  }
\caption{Function $\gamma_{check}$}  
\end{algorithm}

\subsection{Rule Instantiation}
\label{subsec:instantiation}

In this section is described the process of instantiation of a rule.
This process is a lazy one only called when needed.
Since  we only consider safe rules, the instantiation of a rule is in fact the instantiation of its positive body.
In a forward chaining approach, the only rule instantiations of interest are those that lead to a not blocked supported rule or an unblocked weakly supported rule. 
Hence, the rule instantiation is mainly directed by the instantiated atoms already present in the sets \IN. and \MBT..

The algorithm used in the \asperix. solver and described below is inspired by the previous work realized on the \dlv. grounder
\cite{dlvgrounder,DLVinst} which is based on the semi-naive evaluation technique of~\cite{Ullman}. 
The goal is to find a substitution for all the literals of the body of a rule $r$ thanks to the atoms already in \IN., \MBT. or \OUT..
To do this, a partial substitution $\theta$ is built as possible values are found for the variables of the literals of the body of the rule $r$.
It is assumed that the literals $l_1$, $l_2$, \dots, $l_n$ of the body of the rule $r$ are ordered following a list $[l_1,l_2,\dots,l_n]$: $firstLiteral(r)$ (resp. $lastLiteral(r)$) corresponds to $l_1$ (resp. $l_n$) and $previousLiteral(r)$ (resp. $nextLiteral(r)$) corresponds to the literal which precedes (resp. follows) the literal under consideration in the list.
The substitution calculus for a literal $l$ of a rule $r$ is realized thanks to the functions $firstMatch$ and $nextMatch$.
These functions look for a substitution which has not already been computed, i.e. not leading to a substitution for $r$ present in the set $subst(r)$ of all substitutions $\theta$ such that $\theta(r)$ has already been fired.
If the literal $l$ is positive, a substitution such that the substituted atom is in the set \IN. (or $\IN. \cup \MBT.$) is searched.
If the literal is negative, (a) a substitution such that the substituted corresponding atom is in the set \OUT. is searched if the goal is an unblocked rule or (b) the non membership of the substituted atom to the set $\IN.\cup\MBT.$ is checked if the goal is a not blocked rule\footnote{In this case, the body of the rule is ordered in such a way that negative literals appear after the positive literals containing their variables.}.

In the functions $firstMatch$ and $nextMatch$ which follow, the parameter $\gamma$  shows if an unblocked weakly supported or not blocked supported rule is looked for.

\begin{itemize}
 \item $firstMatch(l, \theta, \gamma, IN, MBT, OUT, subst)$ is a function which searches for the first possible substitution for a literal $l$ w.r.t. the sets \IN., \MBT. and \OUT., selection criterion $\gamma$ (unblocked weakly supported or applicable rule) and the current partial substitution $\theta$.
$firstMatch$ returns true and updates the partial substitution $\theta$ in case of success.
Otherwise, the function returns false.

 \item $nextMatch(l, \theta, \gamma, IN, MBT, OUT, subst)$ is a function which searches for the next possible substitution for literal $l$ given the already realized substitutions. 
\end{itemize}

For a rule $r$, a {\em free variable} of a literal $l$ is an occurrence of a variable $X$ such that it is its first occurrence in the body of $r$ when starting traversing the literal $l$.
In other words, no other literal which precedes $l$ in the body of $r$ contains an occurrence of the variable $X$. 
During the instantiation of a rule, a possible substitution is sought for all the free variables of every traversed literal and the substitutions of the previously calculated variables are kept.
If a literal has no free variable, the validity of the substitution w.r.t. the selection criterion $\gamma$ is checked (i.e. the substituted corresponding atom $\theta(at(l))$ is in \IN. or $\IN.  \cup \MBT.$ if $l \in \corpspos{r}$ and  $\theta(at(l))$ is in \OUT. or $\theta(at(l))$ is not in \IN. if $l \in \corpsneg{r}$).

\begin{example}
\label{ex:freeVariables}
Let \intextrule{a(X,Y,Z)}{b(X,Y), c(X,Y), d(X,Z)} be a rule.
The ordered list of the body of the rule is $[l_1=b(X,Y), l_2=c(X,Y), l_3=d(X,Z)]$ with:
  \begin{itemize}
      \item $freeVariables(l_1) = \{ X, Y \}$
      \item $freeVariables(l_2) = \emptyset$ 
      \item $freeVariables(l_3) = \{ Z \}$.
  \end{itemize}

\end{example} 

\begin{algorithm}[!ht]
\label{algo:instantiation}

  \FuncSty{Function instantiateRule($r,~ \gamma,~ IN,~MBT,~ OUT, ~ subst$)}\;
  $\theta \leftarrow lastSubstitution(r)$\; \label{ln:last-subst}
  \eIf{$\theta = \emptyset$}
  {\label{ln:empty-theta}
    \tcc{Searching for the first possible substitution of first literal}
    $l \leftarrow firstLiteral(r)$\;
    $matchFound \leftarrow firstMatch(l, \theta, \gamma, IN, MBT, OUT,subst)$\;
  }
  {\label{ln:nonempty-theta}
    \tcc{Searching for the next possible substitution of last literal}
    $l \leftarrow lastLiteral(r)$\;
    $\theta \leftarrow \theta \backslash freeVariableSubstitutions(l)$\;
    $matchFound \leftarrow nextMatch(l, \theta, \gamma, IN, MBT, OUT,subst)$\;
  }
  \While{\true.}
  {
    \eIf{$matchFound$}
    {\label{ln:matchfound}
      \eIf {$l \neq lastLiteral(r)$}
      {
	$l \leftarrow nextLiteral(r)$\;
	$matchfound \leftarrow firstMatch(l, \theta, \gamma, IN, MBT, OUT,subst)$\
      }
      {\label{ln:complete-substitution}
	  \tcc{A complete substitution is found}
	  \Return $\theta$\;
      }
    }
    {
      \tcc{No substitution for literal $l$. Bactrack to previous literal (if any) to find its next possible substitution}
      \eIf{$l \neq firstLiteral(r)$}
      {\label{ln:backtrack}
	$l \leftarrow previousLiteral(r)$\;
	$\theta \leftarrow \theta \backslash freeVariableSubstitutions(l)$\;
	$matchFound \leftarrow nextMatch(l, \theta, \gamma, IN, MBT, OUT,subst)$\;
      }
      {
	\Return \NULL.\; \label{ln:failure}
      }
    }
  }
\caption{$instantiateRule$}

\end{algorithm}

Function $instantiateRule$ of Algorithm~\ref{algo:instantiation} specifies the instantiation principles of a rule for constant sets \IN., \MBT. and \OUT..
This function is initialized with the partial substitution $\theta$ which is the last found substitution (thanks to the function  $lastSubstitution$) for the rule $r$ if any (line~\ref{ln:last-subst}). 
If it is the first attempt for the instantiation of this rule, $\theta$ is empty  (line~\ref{ln:empty-theta}) and the function searches a first substitution for the first literal of the body of the rule $r$ using the function $firstMatch$.
Otherwise, a substitution for $r$ has already been computed (line~\ref{ln:nonempty-theta}), the function searches a new possible substitution for the rule. For this, the function searches the next possible instance of the last literal of the rule $r$ by deleting from $\theta$ the substitutions of the free variables of this literal (thanks to the function  $freeVariableSubstitutions$) and by calling the function $nextMatch$.
During the execution of the main loop, the function first checks if a substitution has been found for the current literal $a$ (line~\ref{ln:matchfound}).
If it is the case, it searches a first substitution for the next literal of the rule body respecting the partial substitution $\theta$.
When all the atoms have been considered, a complete substitution is found (line~\ref{ln:complete-substitution}).
The function returns this substitution.
When the instantiation of a literal fails (i.e. there is no possible substitution for it), the function backtracks on the previous literal (line~\ref{ln:backtrack}) and updates $\theta$ by deleting the substitutions of the free variables of this literal.
Hence the function calls the function $nextMatch$ which searches the next possible instantiation for this literal.
The instantiation of a rule $r$ fails when no more substitution is possible for the first literal (line~\ref{ln:failure}).

Actually, the instantiation algorithm of a rule is slightly more complicated than the Algorithm~\ref{algo:instantiation} since the atoms dynamically added into \IN., \MBT. and \OUT. sets during the answer set computation, called \emph{atoms to propagate}, have to be taken into account: if possible, each substitution has to be computed once and only once.
Hence \asperix. uses a queue called $propagate\_IN$ (resp. $propagate\_MBT$ and  $propagate\_OUT$) which contains the atoms to be added into the set \IN.  (resp. \MBT. and \OUT.).
When all the instances of a rule $r$, for given sets $\IN._0$,  $\MBT._0$ and $\OUT._0$, have been generated, an atom to propagate $a_p$ whose predicate symbol $p$ appears in the body of the rule $r$ is extracted.
Now $I_1 = \langle {\IN._1}, {\MBT._1}{\OUT._1}\rangle$ denotes the mbt partial interpretation obtained by adding $a_p$ into $I_0 = \langle \IN._0,\MBT._0, \OUT._0\rangle$.
The body of the rule is ordered in such a way that the first literals are those whose predicate symbol is the one of the atom to propagate $a_p$
(they are the literals that might unify with $a_p$).
Then these literals whose predicate symbol is $p$ are successively marked and placed at the beginning of the rule.
The marked literal might only take the value  of the atom to propagate $a_p$ whereas the following (non marked) literals might take any values in $I_1$.
Then, if the instantiation of the first literal fails, it is unmarked, the next literal of predicate $p$ becomes the first literal of the rule body and is marked in turn, and the instantiation of the rule is started again.
The unmarked literals might then take any values in $I_0$ (which excludes the values of $a_p$ already used) while the marked literal can only take the value of the atom to propagate, and the non marked literals always take their values  in $I_1$.
If the instantiation of the first literal fails and there is no other literal to be marked, the instantiation of the rule fails.

\begin{example}
\label{ex:instantiation}
Let $r_0$ be a rule and $\IN._0$, $propagate\_IN$ and ${\IN._1}$ be sets of atoms defined as follow:
 \begin{center}
      $	  \begin{array}{l}

				r_0 = \outtextrule{a(X+Y)}{a(X),  b(X,Y),  a(Y)} \\
				\IN._0 = \{b(1,1), b(1,2)\}\\
				propagate\_IN = \{a(1)\}\\
				{\IN._1} = \{b(1,1), b(1,2), a(1)\}\\
	\end{array} $
 \end{center}

 
\begin{table}[ht]

    \centering
    \begin{tabular}{l l l l l r}
    $a(X+Y)$ & $\leftarrow$	 & $a(X),$ & $a(Y),$ & $b(X,Y)$ & \\
    & & \multicolumn{3}{l}{*** first call to \instantiateRule. ***} & \\
(1.1)			  &  & $\bold{a(1)}$ & - & - & ($l_1$ marked)\\
(1.2)			  &  & $\bold{a(1)}$ & $a(1)$ & - & \\
(1.3)			  &  &$\bold{a(1)}$ & $a(1)$ & $b(1,1)$ & $\Rightarrow$ complete instantiation \\
    & & \multicolumn{3}{l}{*** second call to \instantiateRule. ***} & \\
(2.1)			  &  & $\bold{a(1)}$ & $a(1)$ & NO & \\
(2.2)			  &  & $\bold{a(1)}$ & NO & - & \\
(2.3)			  &  &{\bf  NO} & - & - &  $\Rightarrow$ failure \\
(2.4)			  &  & - & $\bold{a(1)}$ & - &  	($l_2$ marked) \\
(2.5)			  &  & NO & $\bold{a(1)}$ & - &  	$\Rightarrow$ failure \\
    \end{tabular}
    \caption{Decomposition of the instantiations of the rule $r_0$ for the atom to propagate $a(1)$ (Example \ref{ex:instantiation})\label{tab:a1}}
\end{table}
 
The atom $a(1)$ has to be propagated by instantiating the rule $r_0$.
Table~\ref{tab:a1} shows the different steps of the instantiation.
The literals to be marked (whose predicate symbol is $a$) of the body of the rule are $a(X)$ and $a(Y)$.
These literals are placed at the beginning of the body of $r_0$ like this: $[l_1= a(X),  l_2= a(Y), l_3= b(X,Y) ]$.
In Table~\ref{tab:a1}, for clarity, the sequence of literals of the rule body is not changed when the marked literal changes. But the marked literal (shown in bold) is processed first, which is the same.
The first attempt for an instantiation begins and for the first time with atom to propagate $a(1)$.
The literal $l_1=a(X)$ is then marked and takes as unique value that of the atom to propagate $a(1)$ ((1.1) Table~\ref{tab:a1}).
Hence, value $1$ is substituted to the variable $X$ in $\theta$. 
Then, the following literal in the body of the rule, $l_2=a(Y)$, becomes the current literal and takes as value the first amongst those into ${\IN._1}$ which is also $a(1)$. 
Hence, value $1$ is substituted to the variable $Y$ in $\theta$ ((1.2) Table~\ref{tab:a1}).
Then the last literal, $l_3= b(X,Y)$, is reached.
This literal has no free variable and the membership into ${\IN._1}$ is simply checked for $b(1,1)$ which is obtained from $b(X,Y)$ by substituting $X$ and $Y$ by the values in $\theta$ ((1.3) Table~\ref{tab:a1}).
There is no more literal to consider then a complete substitution has been found.
The atom of the head $a(X+Y)$ takes the values of the substitution $\theta$. 
Hence, the forward chaining algorithm can add $a(2)$  into the $propagate\_IN$ queue.

Now, during a new instantiation attempt of the rule for the atom to propagate $a(1)$, the function restarts with the last substitution of the rule $\theta = \{ X/1, Y/1 \}$ in order to find a new substitution for the literal  $l_3=b(X,Y)$.
The second attempt for an instantiation begins with atom to propagate $a(1)$ for the second time.
Since $b(X,Y)$ has no free variable, there can be no other substitution than the current one ((2.1) Table~\ref{tab:a1}).
The process then backtracks to the literal $l_2=a(Y)$ which has no other substitution in ${\IN._1}$ ($a(2)$ has been inferred after $a(1)$ and is not into the current set ${\IN._1}$) ((2.2) Table~\ref{tab:a1}).
Since literal $l_1=a(X)$ can only take the value $a(1)$, it also fails ((2.3) Table~\ref{tab:a1}).
Since the last literal has failed, the literal $l_2=a(Y)$ is now marked instead of $a(X)$, and is instantiated with the atom to propagate $a(1)$. Hence, value $1$ is substituted to the variable $Y$ in $\theta$  ((2.4) Table~\ref{tab:a1}).
Literal $l_1=a(X)$ is unmarked and can only  take the values of the atoms of $\IN._0$, thus no substitution is possible.
Hence the algorithm fails on the first literal ((2.5) Table~\ref{tab:a1}).
Since there is no more literal to be marked, the rule instantiation ends by a failure for the atom to propagate $a(1)$.
The sets becomes as follow:

\begin{center}
	 $ \begin{array}{l}
				\IN._0 = \{b(1,1), b(1,2), a(1)\}\\
				propagate\_IN = \{a(2)\}\\
				{\IN._1} = \{b(1,1), b(1,2), a(1),a(2)\}\\
	\end{array} $
 \end{center}
 
\begin{table}[ht]

    \centering
    \begin{tabular}{l l l l l r}
    $a(X+Y)$ & $\leftarrow$	 & $a(X),$ & $a(Y),$ & $b(X,Y)$ & \\
    & & \multicolumn{3}{l}{*** third call to \instantiateRule. ***} & \\
(1.1)			  &  & $\bold{a(2)}$ & - & - &  ($l_1$ marked)\\
(1.2)			  &  & $\bold{a(2)}$ & $a(1)$ & - & \\
(1.3)			  &  & $\bold{a(2)}$ & $a(1)$ & NO & \\
(1.4)			  &  & $\bold{a(2)}$ & $a(2)$ & - & \\
(1.5)			  &  & $\bold{a(2)}$ & $a(2)$ & NO & \\
(1.6)			  &  & $\bold{a(2)}$ & NO & - & \\
(1.7)			  &  &{\bf  NO} & - & - &  $\Rightarrow$  failure\\
(1.8)			  &  & - & $\bold{a(2)}$ & - &  ($l_2$ marked)\\
(1.9)			  &  & $a(1)$ & $\bold{a(2)}$ & - & \\
(1.10)			  &  & $a(1)$ & $\bold{a(2)}$ & $b(1,2)$ & $\Rightarrow$ complete instantiation \\
    & & \multicolumn{3}{l}{*** fourth call to \instantiateRule. ***} & \\
(2.1)			  &  & $a(1)$ & $\bold{a(2)}$ & NO & \\
(2.2)			  &  & NO & $\bold{a(2)}$ & - & \\
(2.3)			  &  & - & {\bf NO} & - & $\Rightarrow$  failure\\
    \end{tabular}
    \caption{Decomposition of the instantiations of the rule $r_0$ for the atom to propagate $a(2)$ (Example \ref{ex:instantiation} continued)\label{tab:a2}}
\end{table}

The next atom $a(2)$ is extracted from the queue to propagate.
The third attempt for an instantiation of $r_0$  begins with atom to propagate $a(2)$ for the first time.
Table~\ref{tab:a2} shows the different steps of the instantiation.
The literals $a(X)$ and $a(Y)$ are again to be marked.
The rule instantiation is restarted with the literal $l_1=a(X)$ which is the marked literal.
The variable $X$ is substituted by the value $2$ since the only allowed value is that of the atom to propagate $a(2)$ ((1.1) Table~\ref{tab:a2}).
The current literal is now $l_2=a(Y)$ where $Y$ is substituted by the value $1$ since $a(1)$ is into ${\IN._1}$ ((1.2) Table~\ref{tab:a2}).
The literal $b(X,Y)$ has no free variable and since the atom $b(2,1)$ which respects the substitution $\theta = \{ X/2, Y/1 \}$ is not in ${\IN._1}$, the literal $b(X,Y)$ has no possible substitution ((1.3) Table~\ref{tab:a2}).
Then a new instantiation for  $l_2=a(Y)$ is sought: its next possible value is $2$ (since $a(2)$ is in ${\IN._1}$) ((1.4) Table~\ref{tab:a2}).
Again, since the atom $b(2,2)$ which respects the substitution $\theta = \{ X/2, Y/2 \}$ is not in  ${\IN._1}$, the literal $b(X,Y)$ has no possible substitution ((1.5) Table~\ref{tab:a2}).
The process backtracks to the literal $l_2=a(Y)$ which has no possible value ((1.6) Table~\ref{tab:a2}).
Hence, the process backtracks to the literal $l_1=a(X)$ which has no possible value since the only possible value was that of the atom to propagate $a(2)$ ((1.7) Table~\ref{tab:a2}).
Since the first literal has failed, the process restarts by marking the second literal $a(Y)$ (and unmarking the first $a(X)$).
The marked literal $l_2=a(Y)$ is processed first, it substitutes $Y$ by the value $2$ of the atom to propagate $a(2)$ ((1.8) Table~\ref{tab:a2}).
The unmarked literal  $l_1=a(X)$ may only take its values into $\IN._0$.
The variable $X$ is then substituted by the value $1$ ((1.9) Table~\ref{tab:a2}).
The literal $l_3=b(X,Y)$ has no free variable and since $b(1,2)$ which respects the substitution $\theta = \{ X/1, Y/2 \}$ is in ${\IN._1}$  a complete substitution is found ((1.10) Table~\ref{tab:a2}).
The atom $a(X+Y)$ of the head takes then the value of the substitution $\theta$.
Hence, the forward chaining algorithm can add $a(3)$ into $\IN._1$ and into $propagate\_IN$.

Then, during a new instantiation attempt of the rule $r_0$, the atom to propagate is still $a(2)$.
The process restarts from the last substitution $\theta = \{ X/1, Y/2 \}$ and search for a new substitution for the literal $l_3=b(X,Y)$.
A fourth attempt for an instantiation begins with atom to propagate $a(2)$ for the second time.
Since $b(X,Y)$ has no free variable, there can be no other substitution than the current one ((2.1) Table~\ref{tab:a2}).
The process then backtracks to the literal $l_1=a(X)$ that has no other substitution since the only possible values are those from $\IN._0$ (then neither the atom to propagate $a(2)$ nor $a(3)$ appeared after $a(2)$ are possible) ((2.2) Table~\ref{tab:a2}).
The literal $l_2=a(Y)$  also fails since the marked literal only accepts the value of the atom to propagate $a(2)$ ((2.3) Table~\ref{tab:a2}).
Since there is no more literal to be marked, the instantiation of the failing rule ends for this atom to propagate.
The process continues with the atom $a(3)$ which also leads to a failure.


\end{example}

\subsection{\asperix. language}
\label{subsec:discussion}
The core language of \asperix. is that of normal logic programs~\cite{gellif88b} with function symbols and true (or strong)  negation without inconsistent answer set.
\asperix. also provides dedicated treatment of lists with built-in predicates, as in \dlvcomplex.~\cite{dlvcomplex}, an extension of \dlv. with lists and sets.
On the other side, \asperix. does not provide aggregate atoms and optimization statements~\cite{DBLP:journals/tkde/BuccafurriLR00} which are accepted by the main current systems. 

One of the important issues in ASP is the treatment of function symbols.
Uninterpreted function symbols are important because they enable representation of recursive structures such as lists and trees.
But reasoning becomes undecidable if no restriction is enforced.
A lot of work has been made for identifying program classes for which reasoning is decidable~\cite{acfil2011-nonmon30,DBLP:journals/tplp/AlvianoFL10,Calimeri:2011:FRP:2350124.2350130,Lierler:2009:OMD:1604371.1604401,DBLP:journals/tplp/BaseliceB10,DBLP:journals/tplp/GrecoMT13}.

The inherent difficulty with functions in
general  (and arithmetic in particular)  in the framework of ASP is that it makes the Herbrand universe
infinite in whole generality. 
ASP grounders \lparse. ~\cite{lparse} and versions up to 3.0 of \gringo. ~\cite{DBLP:conf/lpnmr/GebserST07} 
accept programs respecting some syntactic domain restrictions 
and are able to deal with some restricted versions of functions.  

\dlv. grounder~\cite{dlvgrounder} and \gringo. (since version 3.0)~\cite{gringo}  only require programs to be safe and can deal with all programs having
a finite instantiation.
\dlv. guarantees  finite instantiation for finitely ground programs but membership in this class is not decidable.
It integrates a Finite Checker module which can check if a program belongs to a sub-class of finitely ground programs (argument-restricted programs).
For programs that are not member of this sub-class, answer sets can be computed without preliminary check but ending is not guaranteed.

\asperix. can deal with these programs and with some other programs whose instantiation is infinite but whose answer sets are finite.
For example, the program~$P_{\ref{ex:exInut}a}$ from Section~\ref{sec:intro} is not finitely ground:
intelligent instantiation of the program must be finite to be finitely ground.
The key points of intelligent instantiation  are that rules are instantiated with atoms appearing in head of rules of the program, 
and simplifications are performed relatively to facts and rule heads of preceding components of the dependency graph.
In example~$P_{\ref{ex:exInut}a}$, choice between $a$ and $b$ makes both possible for the grounder, 
and constraint has no effect on intelligent instantiation of the program.
Thus, the grounding of rules from~$P_{\ref{ex:exInut}a}$ will be the same with or without the constraint \intextrule{}{a}: infinite in both cases.
\asperix. halts on~$P_{\ref{ex:exInut}a}$ and  is thus able to halt on non finitely ground programs
but it is not able to verify in advance if answer sets are finite or not, and thus if computation will end or not.
Nevertheless, ending can be guaranteed by means of command-line options specifying the maximum allowed nesting level for functional terms and
the biggest admissible integer (\dlv. grounder provides similar possibilities).
These restrictions ensure that our computations always converge to an answer set if it exists.
Formalizing the class of programs for which \asperix. halts will be the subject of a forthcoming work.

\section{Experimental results}
\label{sec:experiments}

Following Algorithm~\ref{algo:solve} of Section~\ref{subsec:core},
the solver called \asperix. has been implemented in C++ and is available at
\url{http://www.info.univ-angers.fr/pub/claire/asperix}. 

%
There are two other ASP systems, 
\gasp.~\cite{gasp} and \omiga.~\cite{omiga12}, that
realize the grounding of the program during the search of an answer set. 

\gasp. is an implementation in Prolog and Constraint Logic
Programming over finite domains of the notion of computation 
(see Section  \ref{subsec:computation}).
The main ideas are the same as those of \asperix..
Notable differences are the following.
Well founded consequences of the program are computed first.
Then propagation is close to ours. \gasp. does not deal with must-be-true atoms but two special cases of propagation, not treated by \asperix., are implemented:
(a) if the head of a rule is known to be in OUT set and the body of the rule is satisfied except for one positive literal, then this literal must be false (added to OUT)
and (b) if, for some undefined atom $a$,  there is no applicable rule whose head is $a$, then $a$ can be added to OUT.
For each rule, instantiation and propagation are realized by building and solving a CSP that determines atoms derivable from the rule.
Representation of interpretations uses Finite Domain Sets, such a data structure is efficient  to represent compactly intervals
but it need to code tuples (instances of predicates) by integers (very big integers if domain is large and arity of predicate too).
This representation impose the set of ground terms of the program to be finite and thus function symbols are excluded.
On the other hand, \gasp. supports some cardinality constraints. 
To our knowledge, \gasp. remained at the prototype stage and is no longer developed.

\omiga. is implemented in Java.
Functional symbols (of non-zero arity) are not supported.
Principles of propagation and choice are the same as those of \asperix. but implementation uses Rete algorithm for improving the speed of propagation.
First order rules are represented by a Rete network.
Each node represents a literal (or a set of literals) from the body of a rule or the atom of the head of a rule.
It stores all instances of the node that are true w.r.t. current partial interpretation.
Thereby all partial instantiations of rules are stored in the network.
This lead to an efficient propagation regarding computation time, but memory space is sacrified.
Dependency graph and solved predicates seems to be treated in a similar manner to that of \asperix..
Current version~\cite{omiga13} uses must-be-true propagation and tries to introduce methods for conflict-driven learning of non-ground rules:
when a constraint is violated, a new constraint is built by unfolding of rules whose firing contributes to the conflict.
This learned constraint is then transformed into special rules so as to be used for propagation.

In the following we give some results of evaluation of \asperixv. 
highlighting its adequacy to some particular problems.
It is compared with 
\clingo. (composed by \gringov. and \claspv.) ~\cite{gringo,clasp}, 
\dlvv.~\cite{dlv},
\gaspv.~\cite{gasp}
 and \omigav.~\cite{omiga12}. Version without learning is used for \omiga. because learning lowers its performances.
 All the systems have been run on an Intel Core i7-3520M PC with 4 cores 
at 2.90GHz and about 4GB RAM, running Linux Ubuntu 12.04 64 bits. 
For each instance of a problem, the memory usage is limited to 3.000MB and computation time to 600 seconds.
 RunLim1.7 is used for these limitations tasks. Tables of results use $OoM$ (resp. $OoT$) to indicate Out of Memory (resp. Out of Time).
Results for \gasp. are only given for the first two examples, because it does not accept  other tested programs.

\paragraph{\bf Schur problem}
The Schur number problem is to partition $N$ numbers into $M$ sets such that all of the sets satisfy: if $x$ and $y$ are assigned to the same set, then $x+y$ is not in the set.
The following program~\cite{gasp} is for $M=3$ sets and $N=4$ numbers.

 \centerline{\(
 P_{Schur-4} = \left\{
   \begin{array}{llll}
\outtextfact{number(1)},& \outtextfact{number(2)},& \outtextfact{number(3)},& \outtextfact{number(4)},\\
\outtextfact{part(1)},& \outtextfact{part(2)},& \outtextfact{part(3)},\\
\multicolumn{4}{l}{\outtextrule{inpart(X,1)}{not~inpart(X,2),~not~inpart(X,3),~number(X)},}\\
\multicolumn{4}{l}{\outtextrule{inpart(X,2)}{not~inpart(X,1),~not~inpart(X,3),~number(X)},}\\
\multicolumn{4}{l}{\outtextrule{inpart(X,3)}{not~inpart(X,1),~not~inpart(X,2),~number(X)},}\\
\multicolumn{4}{l}{\leftarrow ~number(X), ~number(Y), ~part(P),}\\
\multicolumn{4}{l}{~~~~inpart(X,P), ~inpart(Y,P), ~inpart(Z,P),}\\
\multicolumn{4}{l}{~~~~T=Y+1,~X < T, ~Z = X+Y.}\\
   \end{array}
 \right\} 
 \)}
 \noindent

The results are shown in Table \ref{tab:res:ex:Schur} for $M=3$.
$AS$ reports the number of answer sets which are all computed.
For all $N\geq 14$, Schur-$N$ has no answer set.

The program is a typical ``guess and check'' program.
The seach space is expressed by the three rules with $inpart$ as head predicate, and constraint eliminates ``bad choices''.
The grounding of the program is rather small but the search space is large.
The problem is very easy for \clingo. and \dlv. but very hard for \asperix. and \gasp..
Systems using grounding on the fly have to repeat instantiation of the same rules in each branch of the search tree.
Moreover, constraints are not efficiently managed by systems like \asperix.:
it does not use constraints for propagation but only checks if a constraint is violated.
Compared to \asperix., \omiga. performs well for computation time, 
certainly because Rete network improve speed of instantiation (partial instantiations are stored in the network) and the network remains relatively small in such an example.
This example illustrates a large class of programs that \asperix. mismanage:
programs with many choices and little propagation by forward chaining.

\begin{table}[!htbp]
  \centering
  \begin{footnotesize}
  \begin{tabular}{llr|c|c|c|c|c|}
    \hline
 	&        &              & \asperix. & \clingo.  & \dlv. & \omiga. & \gasp.\\
    \hline
    \hline
   $N=1$& $AS=3$  & time in sec  &  $<$0.1   & $<$0.1    & $<$0.1 & $<$0.1  & $<$0.1  \\
        &        & memory in MB &   $<$2.0  & $<$2.0    & $<$2.0 & 20.0    &  4.4 \\
    \hline
   $N=2$& $AS=6$  & time in sec  &  $<$0.1   &$<$0.1     & $<$0.1 & $<$0.1  & $<$0.1  \\
        &        & memory in MB &  $<$2.0   & $<$2.0    & $<$2.0 & 20.0    & 5.2  \\
    \hline
   $N=3$& $AS=18$  & time in sec  &  $<$0.1   & $<$0.1    & $<$0.1 & 0.1     & 0.3  \\
        &        & memory in MB &  $<$2.0   & $<$2.0    & $<$2.0 & 20.0    & 6.3  \\
    \hline
   $N=4$& $AS=30$  & time in sec  &  $<$0.1   & $<$0.1    & $<$0.1 & 0.2     & 1.0  \\
        &        & memory in MB &  $<$2.0   & $<$2.0    & $<$2.0 & 24.0    & 8.3  \\
    \hline
   $N=5$& $AS=66$  & time in sec  & $<$0.1    & $<$0.1    & $<$0.1 & 0.3     & 3.6  \\
        &        & memory in MB & $<$2.0    &  $<$2.0   & $<$2.0 & 50.0    & 8.3  \\
    \hline
   $N=6$& $AS=120$  & time in sec  & $<$0.1    & $<$0.1    & $<$0.1 & 0.4     & 11.0  \\
        &        & memory in MB & $<$2.0    &  $<$2.0   & $<$2.0 & 63.0    & 8.3  \\
    \hline
   $N=7$& $AS=258$  & time in sec  &    0.3    & $<$0.1    & $<$0.1 & 0.6     & 41.0  \\
        &        & memory in MB &    2.0    & $<$2.0    & $<$2.0 & 99.0    & 7.5  \\
    \hline
   $N=8$& $AS=288$  & time in sec  &    1.0    & $<$0.1    & $<$0.1 & 0.8     & 113.0  \\
        &        & memory in MB &    2.0    & $<$2.0    & $<$2.0 & 100.0   &  5.7 \\
    \hline
   $N=9$& $AS=546$  & time in sec  &   3.6     & $<$0.1    & $<$0.1 & 1.3    & 370.0  \\
        &        & memory in MB &   2.0     & $<$2.0    &$<$2.0 &  165.0  &  7.5 \\
    \hline
  $N=10$& $AS=300$  & time in sec  &  11.1     &$<$0.1     & $<$0.1 & 1.9    & OoT  \\
        &        & memory in MB &   2.0     & $<$2.0     & $<$2.0 & 220.0  &  - \\
    \hline
  $N=11$& $AS=186$  & time in sec  &  39.7     & $<$0.1    & $<$0.1 & 2.8    & OoT  \\
        &        & memory in MB &   2.2     & $<$2.0     & $<$2.0 & 290.0   &  - \\
    \hline
  $N=12$& $AS=114$  & time in sec  &  131.0    & $<$0.1    & $<$0.1 & 4.1     & OoT  \\
        &        & memory in MB &   2.2     &  $<$2.0    & $<$2.0 & 290.0        & -  \\
    \hline
  $N=13$& $AS=18$  & time in sec  &  448.0    & $<$0.1    & $<$0.1 & 6.5     & OoT  \\
        &        & memory in MB &   2.2     &  $<$2.0    & $<$2.0 & 285.0        &  - \\
    \hline
  $N=14$& $AS=0$  & time in sec  & OoT    & $<$0.1    & $<$0.1 & 11.0     & OoT  \\
        &        & memory in MB &   -     &  $<$2.0    & $<$2.0 & 287.0        &  - \\
    \hline
  \end{tabular}
  \end{footnotesize}

  \caption{Experimental results for Schur}
  \label{tab:res:ex:Schur}
\end{table}


Conversely, the following examples illustrate problems for which grounding on the fly is well adapted.

\paragraph{\bf Birds problem}

Problem \emph{birds} is a stratified program encoding a taxonomy about flying and non flying birds.
$b$ stands for $bird$, $f$ for $flying$, $nf$ for $non flying$, $p$ for $penguin$, $sp$ for $super penguin$, and $o$ for $ostrich$.

 \centerline{\(
 P_{birds} = \left\{
   \begin{array}{lrr}
     \outtextrule{p(X)}{sp(X)},& \outtextrule{b(X)}{p(X)},& \outtextrule{b(X)}{o(X)}, \\
     \multicolumn{2}{l}{\outtextrule{f(X)}{b(X),~ not~ p(X),~ not~ o(X)},}& ~~\outtextrule{f(X)}{sp(X)},\\
     \multicolumn{2}{l}{\outtextrule{nf(X)}{p(X),~ not~ sp(X)},}& \outtextrule{nf(X)}{o(X)}\\
   \end{array}
 \right\} 
 \)}
 \noindent

 We add to this program the atoms encoding
 $N$ birds with $10 \%$ of ostriches, $20 \%$ of penguins whose half
 of them are super penguins.

%

\begin{figure}[h]
  \centering
  \includegraphics[height=12cm, angle=-90]{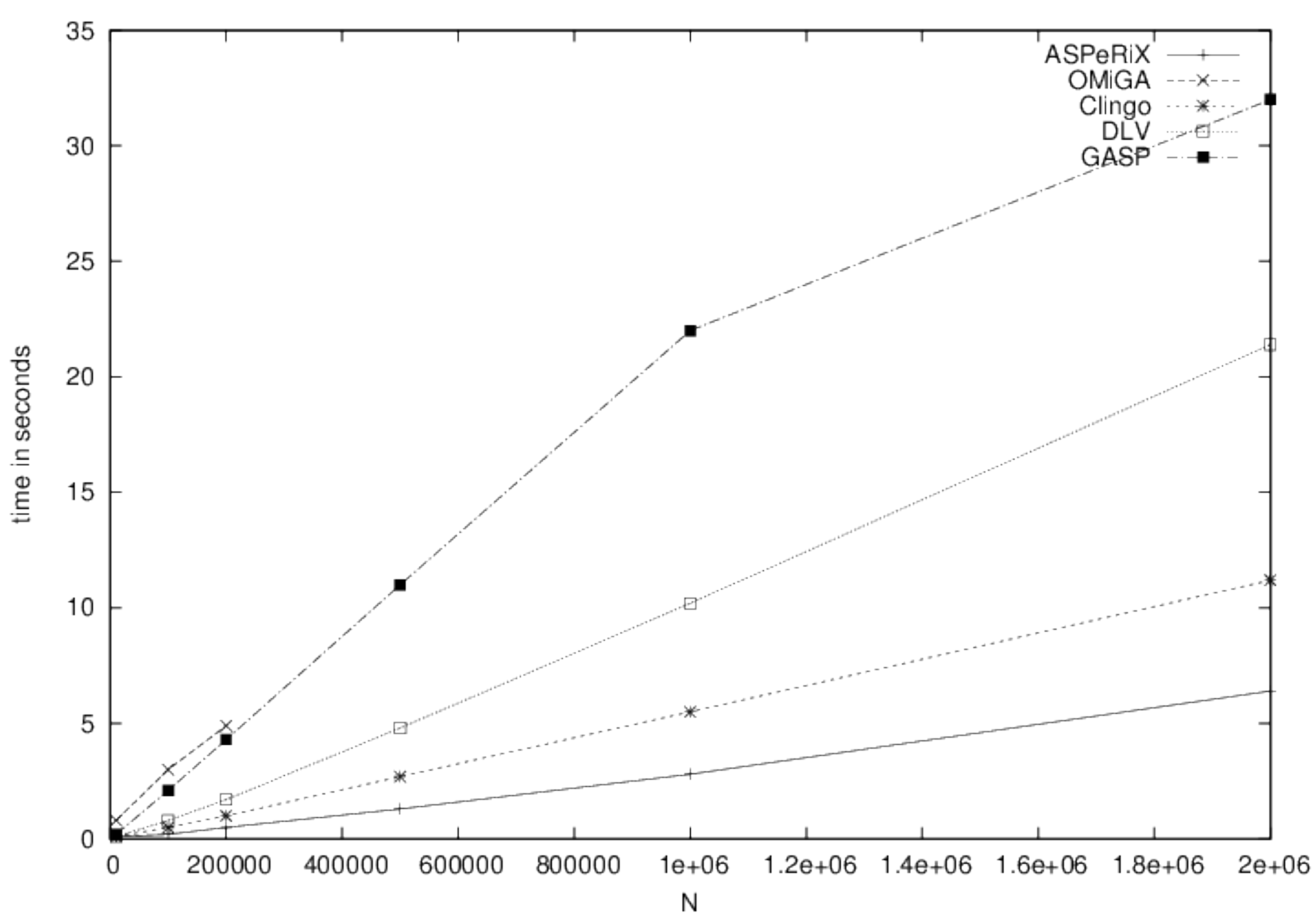}
  \caption{Time for birds}
  \label{fig:res:ex:birds:time}
\end{figure}

\begin{figure}[h]
  \centering
  \includegraphics[height=12cm, angle=-90]{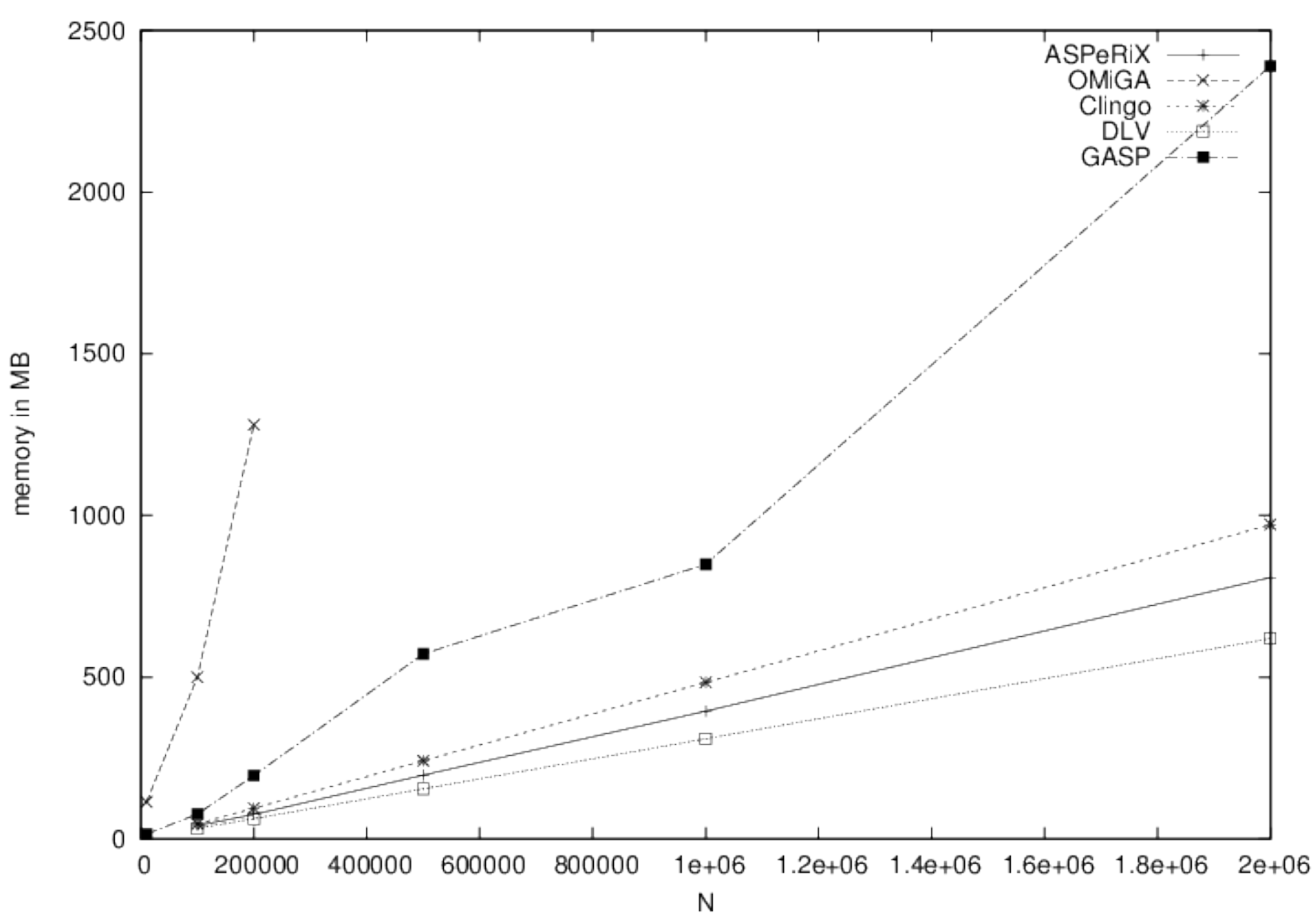}
  \caption{Space for birds}
  \label{fig:res:ex:birds:space}
\end{figure}

The unique answer set of such a program can be computed polynomially.
\asperix. uses only propagation step, without choice point, 
and grounders completely evaluate the program so that the solver has nothing to do.
Experimental results for \emph{birds} are comparable for \asperix., \clingo., \gasp. and \dlv..
\asperix. has the best results for CPU time, and \dlv. for memory usage 
(Figure \ref{fig:res:ex:birds:time} and \ref{fig:res:ex:birds:space}).
For such a problem, the number of instantiated rules must be nearly the same for all systems.
On the other side, \omiga. system uses a very large amount of memory space,
certainly due to the Rete network which is designed to sacrifice memory for increased speed.
Unfortunately, memory gains expected by the first order approach  are lost.

\paragraph{\bf Cutedge problem}

\emph {cutedge} program is proposed in~\cite{omiga12}: given a random graph with 100 vertices
and $N$ edges, each answer set is obtained by deleting an edge and compute some transitive closure
on the remaining edges.

 \centerline{\(
 P_{cutedge} = \left\{
   \begin{array}{l}
      \outtextrule{delete(X,Y)}{edge(X,Y),~not~keep(X,Y)},\\
      \outtextrule{keep(X,Y)}{edge(X,Y),~delete(X1,Y1),~X1 != X},\\
      \outtextrule{keep(X,Y)}{edge(X,Y),~delete(X1,Y1),~Y1 != Y},\\
      \outtextrule{reachable(X,Y)}{keep(X,Y)},\\
      \outtextrule{reachable(X,98)}{reachable(X,Z),~reachable(Z,98)}\\
   \end{array}
 \right\} 
 \)}
 \noindent

\begin{table}[!htbp]
  \centering
  \begin{footnotesize}
  \begin{tabular}{llr|c|c|c|c|c|}
    \hline
	&       &       & \asperix. & \clingo.  & \dlv. & \omiga. \\
    \hline
    \hline
    $N=2.8K$ & $AS=1$ & time  in sec &  $<$0.1  &  21 &  115  &  0.6 \\
		&    & memory in MB &  14.8  &  345 &  103  &  85  \\
    \hline
    & $AS=10$ & time  in sec &  0.2   &  21 &  226  &  1.8 \\
		  &    & memory in MB & 14.8  &  345 &  103  &  200 \\
    \hline
    & $AS=100$ & time  in sec &  2.7  &   32 &  OoT  &  11.4 \\
		&      & memory in MB  & 15.1  &  345 &   -   &  1042 \\
    \hline
    & $AS=500$ & time  in sec &  16    &   78 &   OoT    &  48 \\
		&      & memory in MB  &  16.4  &  345 &  -   &  1050 \\
    \hline
    & $AS=1000$ & time  in sec &  36    &   123 &   OoT   &  84 \\
		&	& memory in MB &  18.0  &  345  &  -   &  1050 \\
    \hline
    & $AS=all$ & time  in sec & 167  &   189 &   OoT    &  165 \\
		 &    & memory in MB &  24.1 &  345  &  -    &  1050  \\
    \hline
    \hline
    $N=4.9K$ & $AS=1$ & time  in sec &  0.1   & 60  & 325  &   1   \\
		&    & memory in MB &  23.7  &  881 & 144  &  177    \\
    \hline
    & $AS=10$ & time  in sec &  0.5   &  63  & OoT  &   3.7   \\
		 &     & memory in MB &  23.8  &  881 &  -   &  425  \\
    \hline
    & $AS=100$ & time  in sec &   5.8   & 94   &   OoT  &  29   \\
		&      & memory in MB  &  24.0   &  881 &  -  &  1100     \\
    \hline
    & $AS=500$ & time  in sec &  30.7  & 228  &  OoT   &   125   \\
		&      & memory in MB  &  25.3  &  881 & -   &   1150     \\
    \hline
    & $AS=1000$ & time  in sec &  67  & 373  &  OoT   &   245   \\
		&	& memory in MB &  27  &  881 & -   &   1190     \\
    \hline
    \hline
     $N=5.9K$ & $AS=1$ & time  in sec &   0.1   & 94   & 465  &   1   \\
		&      & memory in MB &  28.5  &  1167 & 202  &  132  \\
    \hline
    & $AS=10$ & time  in sec & 0.8  & 94    & OoT &  5  \\
		&      & memory in MB & 28.6 &  1168 &  -  & 680 \\
    \hline
    & $AS=100$ & time  in sec &  7.6  &  114  &  OoT   &  41  \\
		&      & memory in MB  & 29.1 &  1168  &  -  &  1135   \\
    \hline
    & $AS=500$ & time  in sec &  42   & 210  &  OoT   &  192   \\
		&      & memory in MB  & 31.3 &  1168 & -   &   1125  \\
    \hline
    & $AS=1000$ & time  in sec &  92    & 316   &  OoT   &  352  \\
		&	& memory in MB &  34.1  &  1168 &  -  &   1132  \\
    \hline
  \end{tabular}
  \end{footnotesize}

  \caption{Experimental results for cutedge}
  \label{tab:res:ex:cutedge}
\end{table}

%
%

Computing each answer set is only based on propagation, and the number of answer sets equals the number of edges.
The number of rules needed to compute all answer sets is proportional to $N^2$ 
while the rule number needed to compute one is proportional to $N$.
But systems with pregrounding phase must generate all ground instances of rules
even if only one answer set is required.
The results are shown in Table \ref{tab:res:ex:cutedge}.
\asperix. has the best results for this program both for CPU time and memory usage.
\omiga. and \clingo. use much more memory and are much slower than \asperix..
As expected, memory usage of \clingo. is independent of the number of answer sets required 
and is close to the square of that used by \asperix..
For its part, \dlv. quickly exceeds the time limit imposed.
\paragraph{\bf Hamiltonian cycle problem}

The program $P_{\ref{ex:hamcirccompgraph}}$ (see Example~\ref{ex:hamcirccompgraph}),
Hamiltonian cycle in a complete graph,
  is another easy problem with a lot of answer sets.
Each answer set is easy to compute but the whole instantiation is huge.
Experiments for the computation of one answer set in a graph with $N$ vertices are represented in Figures \ref{fig:res:ex:hamiltonian:time} and \ref{fig:res:ex:hamiltonian:space}. 
\asperix. performs well on this example whereas \omiga. has time and memory problems 
similar to that of \clingo. and \dlv..
One more time, a simple problem becomes intractable by systems with pregrounding phase 
because they drown it in a lot of useless information so that memory used quickly becomes prohibitive.

%


\begin{figure}[h]
  \centering
  \includegraphics[height=12cm, angle=-90]{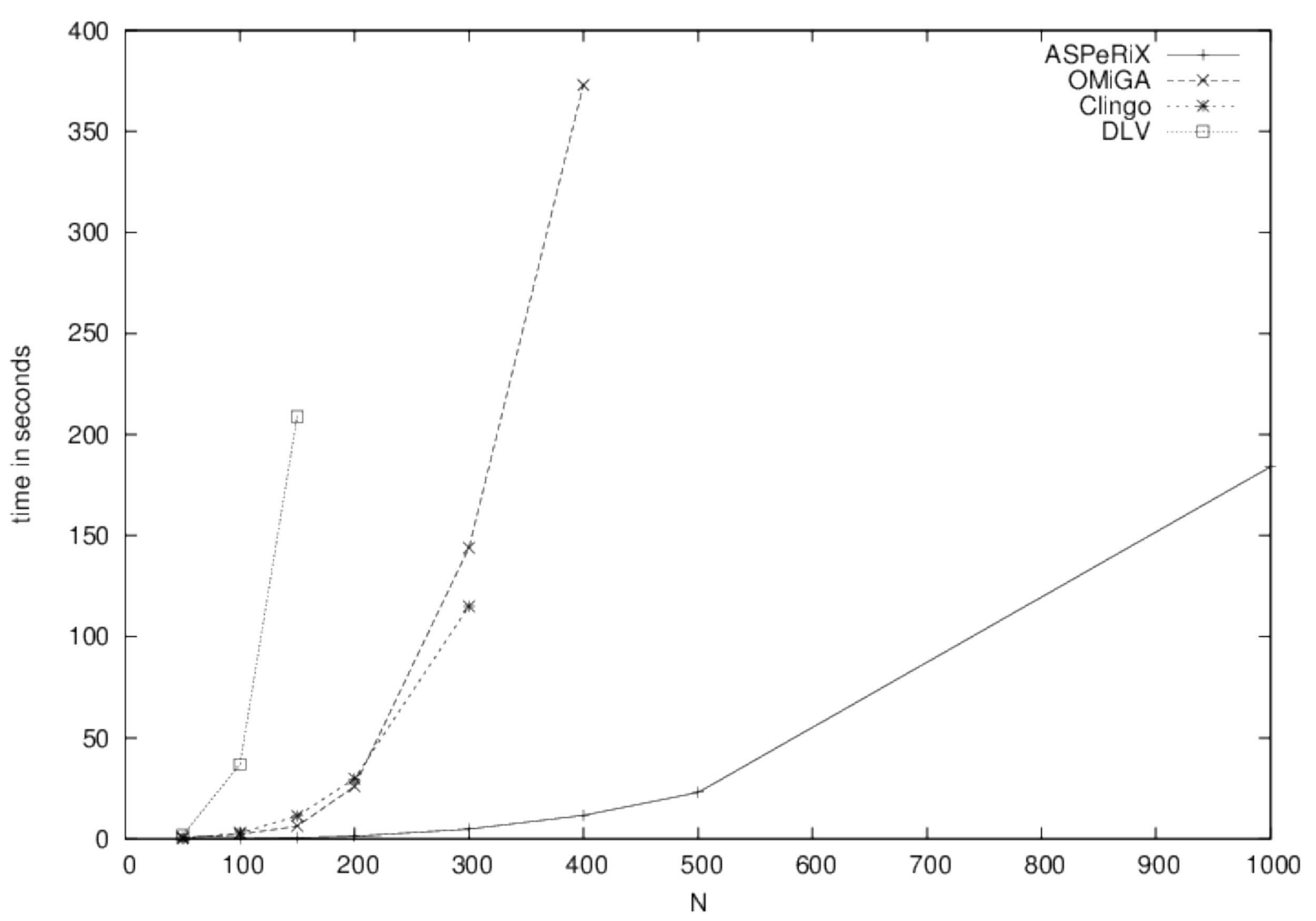}
  \caption{Time for Hamiltonian cycle}
  \label{fig:res:ex:hamiltonian:time}
\end{figure}

\begin{figure}[h]
  \centering
  \includegraphics[height=12cm, angle=-90]{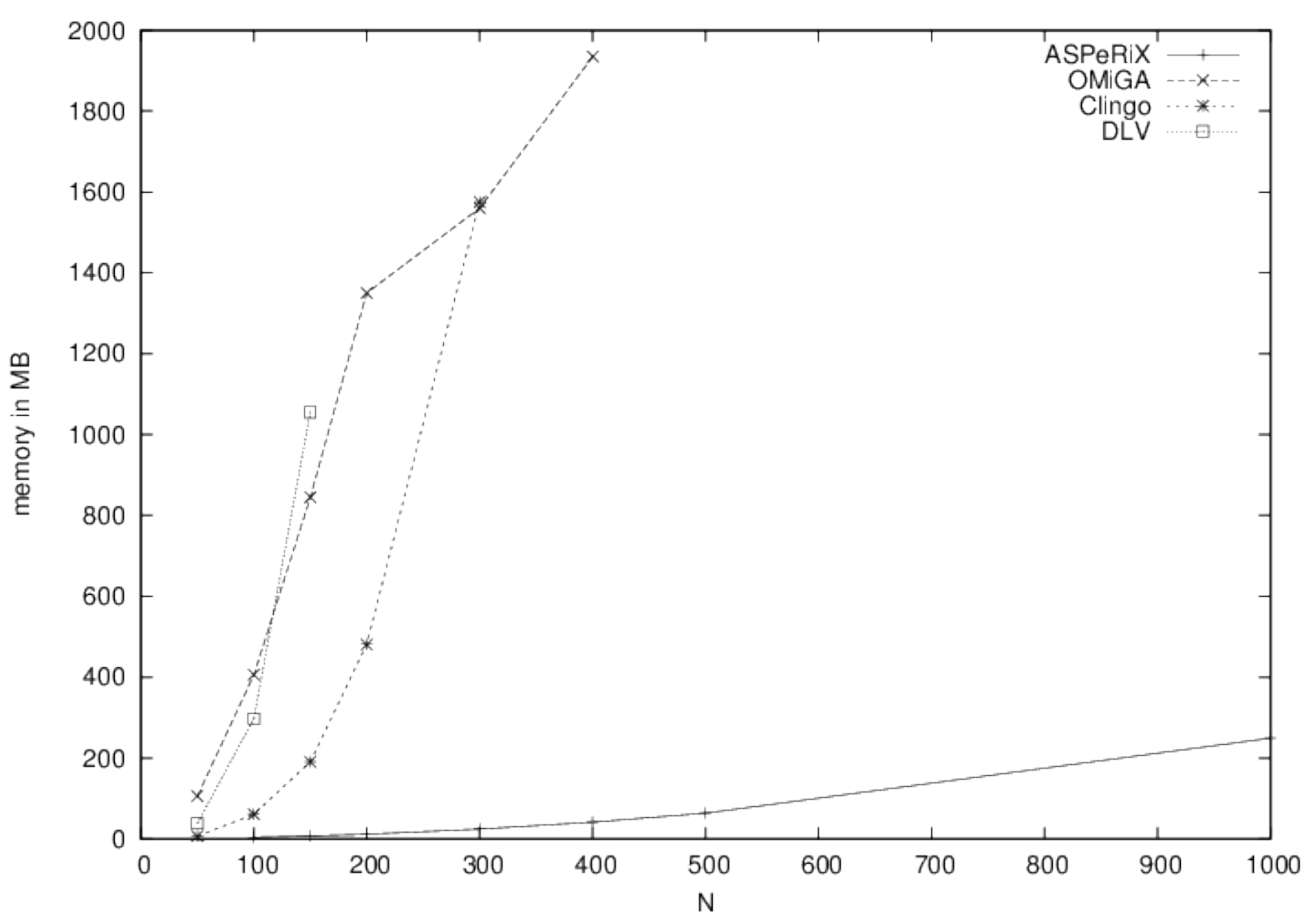}
  \caption{Space for Hamiltonian cycle}
  \label{fig:res:ex:hamiltonian:space}
\end{figure}

\paragraph{\bf Hanoi problem}

 \emph{Hanoi} example illustrates a planning problem 
where the maximum number of allowed steps is given as input.
$NbD$ is the number of disks in the problem and 
$NbM$ is the maximum number of moves that are allowed to move all disks from the first rod to the third.
The least value of $NbM$ is the minimum required to achieve the goal,
then its value is gradually increased to evaluate its impact.
The complete program is given in~\ref{sec:Hanoi} and experimental results are shown in Table \ref{tab:res:ex:hanoi}.
\asperix. performances are (almost) independent of the given number of moves: search, and therefore grounding, are stopped when a solution is found. Conversely, grounders are quickly overwhelmed as they are obliged to fully instantiate the program with all hypothetical (and unnecessary in this case) calculation steps\footnote{
\iclingo.~\cite{iclingo} was created to address this specific problem. 
Some directives are added to the program in order to incrementally instantiate some predicates of the program.
But it does not escape the grounding/solving separation, it only introduces some tools to control the process.
}.
This example cannot be computed by \omiga. due to restrictions on the input language it accepts (function symbols are not supported). 

\begin{table}[!htbp]
  \centering
  \begin{footnotesize}
  \begin{tabular}{llr|c|c|c|}
    \hline
            &          & & \asperix. & \clingo.  & \dlv. \\
    \hline
    \hline
    $NbD=4$ & $NbM=15$ & time in sec &   $<$0.1   &   $<$0.1  &    $<$ 0.1   \\
		      & & memory in MB &   -     &   -    &     -    \\
    \cline{4-6}
	    & $NbM=60$ & time in sec &   $<$0.1    &  0.7   &     0.8     \\
		      & & memory in MB &   -     &   27   &     22     \\
    \cline{4-6}
	    & $NbM=100$ & time in sec &   $<$0.1   & 1.5   &     3.6     \\
		      & & memory in MB &   -     &   54   &     51     \\
    \cline{4-6}
	    & $NbM=500$ & time in sec &   $<$0.1   &  11.6  &     -     \\
		      & & memory in MB &   -     &  327   &     OoM    \\
    \cline{4-6}
	    & $NbM=1000$ & time in sec &   $<$0.1  &   27   &   -      \\
		      & & memory in MB &   -     &  693   &    OoM     \\
    \cline{4-6}
	    & $NbM=2000$ & time in sec &   $<$0.1  &  66   &    -      \\
		      & & memory in MB &   -     &  1523  &    OoM      \\
    \cline{4-6}
	    & $NbM=5000$ & time in sec &   $<$0.1   &    -   &   -       \\
		      & & memory in MB &   -   &  OoM  &    OoM      \\
    \cline{4-6}
	    & $NbM=10000$ & time in sec &  $<$0.1  &  -     &   -       \\
		      & & memory in MB &   -   &   OoM    &     OoM     \\
    \cline{4-6}
	    & $NbM=50000$ & time in sec &   0.1  &  -     &    -     \\
		      & & memory in MB &   12.9   &  OoM    &  OoM        \\
    \cline{4-6}
	    & $NbM=100000$ & time in sec &  0.3 &  -    &    -      \\
		      & & memory in MB  &   23.9  &  OoM    &   OoM       \\
    \hline
    \hline
    $NbD=5$ & $NbM=31$ & time in sec &   0.2    &   0.1   &     0.1   \\
		      & & memory in MB &  3.7   &   6.4   &     4     \\
    \cline{4-6}
	    & $NbM=50$ & time in sec &   0.2    &   0.7   &     0.9     \\
		      & & memory in MB &   3.7  &    28   &     31     \\
    \cline{4-6}
	    & $NbM=100$ & time in sec &   0.2   &   8.2   &     9.4     \\
		      & & memory in MB &   3.8  &   245   &     270     \\
    \cline{4-6}
	    & $NbM=500$ & time in sec &   0.2   &    91   &     -     \\
		      & & memory in MB &   3.8  &  2055   &     OoM    \\
    \cline{4-6}
	    & $NbM=1000$ & time in sec &   0.2  &     -   &    -      \\
		      & & memory in MB &   3.8   &   OoM  &     OoM     \\
    \cline{4-6}
	    & $NbM=5000$ & time in sec &   0.2    &   -    &      -   \\
		      & & memory in MB &   4.8   &   OoM     &    OoM       \\
    \cline{4-6}
	    & $NbM=10000$ & time in sec &   0.2  &   -    &      -    \\
		      & & memory in MB &   5.8    &   OoM    &    OoM       \\
    \cline{4-6}
	    & $NbM=50000$ & time in sec &   0.4  &   -    &      -    \\
		      & & memory in MB &   14.4   &   OoM     &    OoM       \\
    \cline{4-6}
	    & $NbM=100000$ & time in sec &   0.7   &   -   &     -     \\
		      & & memory in MB &   25     &   OoM    &     OoM      \\
    \hline
    \hline
    $NbD=6$ & $NbM=63$ & time in sec &   4.7    &     0.6  &    1.1   \\
		      & & memory in MB &   10.3  &      24  &     29    \\
    \cline{4-6}
	     & $NbM=100$ & time in sec &   4.7   &     9  &     -     \\
		      & & memory in MB &   10.3   &    272  &    OoM     \\
    \cline{4-6}
	    & $NbM=150$ & time in sec &   4.7     &     83  &    -      \\
		      & & memory in MB &   10.3   &   1863   &   OoM       \\
    \cline{4-6}
	    & $NbM=200$ & time in sec &   4.7     &     -   &    -      \\
		      & & memory in MB &   10.3   &    OoM   &   OoM       \\
    \cline{4-6}
	    & $NbM=500$ & time in sec &   4.7     &   -      &    -      \\
		      & & memory in MB &   10.4   &   OoM      &   OoM      \\
    \cline{4-6}
	    & $NbM=1000$ & time in sec &   4.7    &   -      &    -      \\
		      & & memory in MB &   10.5    &  OoM      &    OoM     \\
    \cline{4-6}
	    & $NbM=5000$ & time in sec &   4.7     &   -     &    -      \\
		      & & memory in MB &   11.3    &  OoM      &     OoM     \\
    \cline{4-6}
	    & $NbM=10000$ & time in sec &   4.8     &   -     &     -     \\
		      & & memory in MB &   12.4    &   OoM      &    OoM     \\
    \cline{4-6}
	    & $NbM=50000$ & time in sec &   5   &     -    &     -     \\
		      & & memory in MB &   21     &   OoM      &   OoM       \\
    \cline{4-6}
	    & $NbM=100000$ & time in sec &   5.6   &    -     &     -     \\
		      & & memory in MB &   31.7     &   OoM      &   OoM       \\
    \hline
  \end{tabular}
  \end{footnotesize}
  \caption{Experimental results for Hanoi tower problem}
  \label{tab:res:ex:hanoi}
\end{table}
\paragraph{\bf Three coloring problem}

The program  $P_{\ref{ex:3color}}$ (see Example~\ref{ex:3color}), 3-coloring problem on a graph organized as a
bicycle wheel, 
poses no problem for \clingo. and \dlv. (cf. Table \ref{tab:res:ex:3col}).
But \asperix. and \omiga. have bad results on this example
because they are mismanaging constraints.
Once a vertex is colored, say $red$, 
constraint \intextrule{}{e(V,U), col(V, C), col(U,C)} prohibits coloring adjacent vertices of the same color.
In propositional systems, unit propagation (or equivalent) works well and 
allows to infer that adjacent vertices are not colored red.
But first-order approach does not allow, in general case, to use unit propagation
and thus, constraints are mainly used for verification and not for propagation.
A lot of work remains on these points.
First-order constraints could instead allow more powerful propagation.
Suppose for example a constraint \intextrule{}{p(X,Y), p(Y,Z)}
and $p(1,2)$ is added in \IN. set
then, for all $Z$, $p(2,Z)$ can be excluded at once from current solution,
even if $Z$ values are potentially infinite.
But these opportunities are not exploited yet.

\begin{table}[!htbp]
  \centering
  \begin{footnotesize}
  \begin{tabular}{lll|c|c|c|c|}
    \hline
    &           &  & \asperix. & \clingo.  & \dlv. & \omiga. \\
    \hline
    \hline
    $N=11$ & $AS=1$ & time  in sec & $<$0.1 	&  $<$0.1	& $<$0.1	&  0.3 \\
		 &  & memory in MB	& $<$2	& 1 	& $<$1	&  45  	\\
    \hline
    $N=11$ & $AS=all(6)$  & time  in sec &     3.4 &   $<$0.1 	&  $<$0.1  &  7.7 \\
		 &  & memory in MB &     1.8 & 1 	&   $<$1    &  132  \\
    \hline
    $N=101$ & $AS=1$ & time  in sec  & $<$0.1  & $<$0.1 &  $<$0.1  &  OoT \\
		 & & memory in MB &  3    & 1.5    &  1.1    &  -       \\
    \hline
    $N=101$ & $AS=all(6)$ & time  in sec  &  OoT   & $<$0.1 &  $<$0.1   &  OoT\\
		 & & memory in MB  & -    &  1.8   &   1.4  &   -      \\
    \hline
    $N=501$ & $AS=1$ 	& time  in sec &   1.6   &  $<$0.1 &  $<$0.1 &  OoT \\
		 &  & memory in MB    	&   8.8   &  3.3   &   3.3  &  -  \\
    \hline
    $N=501$ & $AS=all(6)$ & time  in sec  &  OoT   &  $<$0.1 &  0.3  & OoT \\
		 & & memory in MB &  -   & 3.3    &  3.3  &  - \\
    \hline
    $N=1001$ & $AS=1$   & time  in sec  &  13.3  & $<$0.1 & 0.1  &   OoT   \\
		 & & memory in MB    &  15.9  & 5.5    & 5.5    &  -  \\
    \hline
    $N=1001$ & $AS=all(6)$ & time  in sec  &  OoT & $<$0.1 & 1.4  &  OoT   \\
		 & & memory in MB &  - &  5.5   & 5.5  &   -  \\
    \hline
  \end{tabular}
  \end{footnotesize}

  \caption{Experimental time results for 3col}
  \label{tab:res:ex:3col}
\end{table}

To sum up, \asperix. is efficient to deal with stratified programs or 
simple problems whose instantiation is infinite or huge
 but much of which is useless to compute one specific answer set.
On the other hand, the system is not competitive for more  combinatorial problems, 
with a large search space and few solutions,
because propositional methods for propagation, heuristics, learning lemmas
did not apply to the first order case.

\section{Conclusion}
\label{sec:conclusion}
In this paper, we have presented the \asperix. approach to answer set computation.
Our methodology deals with first order rules following a forward chaining
with grounding process realized on the fly and has been implemented in
the ASP solver \asperix..
This paper is the first comprehensive document in which a survey of the 
important techniques relevant to our approach is presented.

Starting from a short description of state-of-the-art ASP working principle, we have  presented by many 
examples the main motivation of our approach: escaping the bottleneck of the preliminary 
phase of grounding in which many state-of-the-art systems fall.
After a presentation of the theoretical foundations of ASP,
we have described by an \asperix. computation our first order forward chaining approach 
for answer set computing and have established the soundness and completeness of this calculus w.r.t. the semantics of ASP
(Proofs are reported in \ref{sec:proofs}).
We have then described in details the main algorithms of \asperix. and 
particularly those which realize the selection of the first order rules to be instantiated and 
applied  according to the current answer set in construction.

Our methodology allows very
good performances for definite and stratified programs.
It outperforms systems with a pregrounding phase for programs with large grounding 
but much of it is unnecessary to solve the problem.
On the other side, performances quickly degrade for combinatorial problems with large search spaces, 
especially if forward chaining propagation can not be exploited.

We have shown that our approach escapes the bottleneck of the preliminary 
phase of grounding that is the only difficulty for some classes of programs.
A direct consequence of our new approach is that the use of symbolic
functions in general and arithmetic calculus in particular inside ASP
is greatly facilitated. 


The forward chaining with the grounding process realized on the fly as an operational semantics 
emphasizes the programming aspect of ASP in which the answer set is not only the result of a black
box but the result of a process that may be followed.
This is interesting when dealing with knowledge coming from the web and expressed in description
logic since the structure of information uses rules that are chained ones with the others (whereas this
is not always the case for a program encoding a combinatorial problem).
Moreover, when dealing with knowledge expressed in description logic, one important
issue is the ability to query the knowledge base. The grounding process realized on the fly
will then allow to focus only on the rules useful to find an answer to the query.
For this category of programs, we think that our approach may be of great interest.

Furthermore, computing the answer sets of a program is a fundamental goal but not an exclusive one. 
Debugging a program, controlling its behavior, introducing in it some features coming from other programming languages may be of great interest for ASP. 
We think that our methodology of answer set computing, guided by the rules of the program, is the good starting point 
towards these new goals.

The \asperix. project is still in progress.
Improvements at the algorithmic level are underway by the development and implementation 
of backjumping and clause learning techniques.
On the other hand, we plan to fully respect the core language ASP~\cite{ASPcompetition2014} by introducing, among others,  minimization / maximization
 and aggregates
and extend it by introducing existentially quantified variables in multi-head rules to encode fragments 
of Description Logics which are logical formalisms for ontologies and the Semantic Web.

\section*{Tribute}

In memory of the late Pascal Nicolas who was at the origin of this work. He sadly 
passed away in 2010 but his enthusiasm, his passion for research and his 
great humanity are still with us.

\bibliography{ref}

\appendix
\section{Hanoi example}
\label{sec:Hanoi}
The following ASP program is the $Hanoi$ example with 4 discs.\\

{\ttfamily
\%------ Initial settings

number\_of\_moves(10000).

largest\_disc(4).\\

\%------ Initial state

initial\_state(towers(l(4,l(3,l(2,l(1,nil)))),nil,nil)).\\

\% ------ Goal state

goal(towers(nil, nil, l(4,l(3,l(2,l(1,nil)))))).\\

\% ------ all discs involved ------

disc(1..4).\\

\% ------ legal stacks ------

legalStack(nil).

legalStack(l(T,nil)) :- disc(T).

legalStack(l(T,l(T1,S))) :- legalStack(l(T1,S)), disc(T), T > T1.\\

\% ------ possible moves ------

possible\_state(0,towers(S1,S2,S3))

\hspace{1 cm}:- initial\_state(towers(S1,S2,S3)),

\hspace{1 cm}legalStack(S1), legalStack(S2), legalStack(S3).

possible\_state(I,towers(S1,S2,S3))

\hspace{1 cm}:- possible\_move(I,T,towers(S1,S2,S3)).\\

\% From stack one to stack two.

possible\_move(J,towers(l(X,S1),S2,S3),towers(S1,l(X,S2),S3))

\hspace{1 cm}:-  possible\_state(I,towers(l(X,S1),S2,S3)),

\hspace{1 cm}number\_of\_moves(N), I<=N, legalStack(l(X,S2)), J=I+1, not ok(I).\\

\% From stack one to stack three.

possible\_move(J,towers(l(X,S1),S2,S3),towers(S1,S2,l(X,S3)))

\hspace{1 cm}:-  possible\_state(I,towers(l(X,S1),S2,S3)),

\hspace{1 cm}number\_of\_moves(N), I<=N, legalStack(l(X,S3)), J=I+1, not ok(I).\\

\% From stack two to stack one.

possible\_move(J,towers(S1,l(X,S2),S3),towers(l(X,S1),S2,S3))

\hspace{1 cm}:-  possible\_state(I,towers(S1,l(X,S2),S3)),

\hspace{1 cm}number\_of\_moves(N), I<=N, legalStack(l(X,S1)), J=I+1, not ok(I).\\

\% From stack two to stack three.

possible\_move(J,towers(S1,l(X,S2),S3),towers(S1,S2,l(X,S3)))

\hspace{1 cm}:-  possible\_state(I,towers(S1,l(X,S2),S3)),

\hspace{1 cm}number\_of\_moves(N), I<=N, legalStack(l(X,S3)), J=I+1, not ok(I).\\

\% From stack three to stack one.

possible\_move(J,towers(S1,S2,l(X,S3)),towers(l(X,S1),S2,S3))

\hspace{1 cm}:-  possible\_state(I,towers(S1,S2,l(X,S3))),

\hspace{1 cm}number\_of\_moves(N), I<=N, legalStack(l(X,S1)), J=I+1, not ok(I).\\

\% From stack three to stack two.

possible\_move(J,towers(S1,S2,l(X,S3)),towers(S1,l(X,S2),S3))

\hspace{1 cm}:-  possible\_state(I,towers(S1,S2,l(X,S3))),

\hspace{1 cm}number\_of\_moves(N), I<=N, legalStack(l(X,S2)), J=I+1, not ok(I).\\

\%------ actual moves ------

\% a solution exists if and only if there is a "possible\_move"

\% leading to the goal.

\% in this case, starting from the goal, we proceed backward

\% to the initial state to single out the full set of moves.\\

\% Choose from the possible moves.

move(I,towers(S1,S2,S3))

\hspace{1 cm}:- goal(towers(S1,S2,S3)), possible\_state(I,towers(S1,S2,S3)).

ok(I) :- move(I,towers(S1,S2,S3)), goal(towers(S1,S2,S3)).

win :- ok(I).

:- not win.\\

move(J,towers(S1,S2,S3))

\hspace{1 cm}:- move(I,towers(A1,A2,A3)),

\hspace{1 cm}possible\_move(I,towers(S1,S2,S3),towers(A1,A2,A3)), J=I-1,

\hspace{1 cm}not nomove(J,towers(S1,S2,S3)).\\

nomove(J,towers(S1,S2,S3))

\hspace{1 cm}:- move(I,towers(A1,A2,A3)),

\hspace{1 cm}possible\_move(I,towers(S1,S2,S3),towers(A1,A2,A3)), J=I-1,

\hspace{1 cm}not move(J,towers(S1,S2,S3)).\\

\%------ precisely one move at each step ------

moveStepI(I) :- move(I,T).\\

:- legalMoveNumber(I), ok(J), I<J, not moveStepI(I).\\

:- legalMoveNumber(I), move(I,T1), move(I,T2), T1!=T2.\\

legalMoveNumber(0).\\

legalMoveNumber(K)

\hspace{1 cm}:- legalMoveNumber(I), number\_of\_moves(J), I < J, K=I+1.\\

\#hide.

\#show move/2.
}

\section{Proofs}
\label{sec:proofs}
%

\subsection{Proof of Theorem \ref{the:AnswersetGR}}
\label{subsec:proofGREnrAS}
\begin{proof} (of Theorem \ref{the:AnswersetGR})
\label{proof:GREnrAS}
Let $P$ be a normal logic program and $X \subseteq  \mathcal{A}$. Let us note first that if $GR_P(X)$
is grounded then $Cn({GR_P(X)}^\emptyset) = \HEAD.(GR_P(X))$.

If $X$ is an answer set of
$P$ then, by a theorem from (Konczak et al. 2006), $GR_P(X)$ is grounded and, also according
to (Konczak et al. 2006), $X = Cn({GR_P(X)}^\emptyset)$. Since $Cn({GR_P(X)}^\emptyset) = \HEAD.(GR_P(X))$,
we can conclude $X = \HEAD.(GR_P(X))$.

Let us now suppose that $X = \HEAD.(GR_P(X))$ and $GR_P(X)$ is grounded. We have 
$Cn({GR_P(X)}^\emptyset) = \HEAD.(GR_P(X))$, then $X = Cn({GR_P(X)}^\emptyset)$ and,
according to (Konczak et al. 2006), $X$ is an answer set of $P$.

\end{proof}

\subsection{Proof of Theorem \ref{the:Aspcomp}}
\label{subsec:proofAspcomp}
We first give some material needed in the proof. 
Auxiliary Lemma~\ref{lem:Asperixenumeration} is used in the proof of Lemma~\ref{lem:Asperixcomp1}. Lemmas \ref{lem:Asperixcomp1} and \ref{lem:Asperixcomp2} establish completeness and correctness.

Lemma~\ref{lem:Asperixenumeration} shows that the generating rules of a program can be ordered so as to correspond to the order of application of rules in an \asperix.
 computation.
 Condition (1) says that a rule used at step $i$ is supported at this step.
 Condition (2) says that if a rule is a member of $\Delta_{pro}$ at step $i$ but is used at a later stage $j$, then all rules used at steps between $i$ and $j$ are members of $\Delta_{pro}$ at step $i$. In other words, condition (2) says that propagation is entirely completed before making a choice.

\begin{lemme}
\label{lem:Asperixenumeration}
Let $P$ be a normal logic program and $X$ be  an answer set of $P$.  
Then,  there exists an enumeration ${\langle r_i \rangle}_{i \in [1..n]}$ of $GR_P(X)$, the set of generating rules of $X$, such that 
for all $i \in [1..n]$ the following two conditions are satisfied:
\begin{enumerate}[label=(\arabic*)]
\item $\BODY.^+(r_i) \subseteq \HEAD.(\{r_k \mid k<i\})$
\item for all $j>i$, if $\BODY.^+(r_j) \subseteq \HEAD.(\{r_k \mid k<i\})$ and $\BODY.^-(r_j) \subseteq \BODY.^-(\{r_k \mid k<i\})$
	then $\BODY.^-(r_i) \subseteq \BODY.^-(\{r_k \mid k<i\})$.
\end{enumerate}
\end{lemme}

\begin{proof} (of Lemma~\ref{lem:Asperixenumeration})
\label{proof:Asperixenumeration}
Let $P$ be a normal logic program and $X$ be  an answer set of $P$.  
By Theorem \ref{the:AnswersetGR},  there exists an enumeration ${\langle r_i \rangle}_{i \in [1..n]}$ of $GR_P(X)$ such that 
 $\forall i \in [1..n]$, $ \BODY.^+(r_i) \subseteq \HEAD.(\{r_k \mid k<i\})$, i.e. such that condition (1) is satisfied.
This enumeration can be recursively modified in the following way in order to verify condition (2).
For each $i \in  [1..n]$, if $r_i$ satisfies (2) then $r_i$ remains at rank $i$, else there exists $r_j$ with $j>i$ that falsifies condition (2). 
In this last case, it suffices to swap the two rules in the enumeration to satisfy condition (2) at rank $i$.

\end{proof}

\noindent {\bf Notation.} If $P$ is a normal logic program and ${\langle R_i, \langle IN_i, OUT_i\rangle \rangle}_{i=0}^{\infty}$ is  
a sequence of ground rule sets $R_i$ and partial interpretations $\langle IN_i,OUT_i \rangle$,
then $\Delta_{pro}^i$ denotes $\Delta_{pro}(P,\langle IN_{i}, OUT_{i}\rangle, R_i)$
and  $\Delta_{cho}^i$ denotes $\Delta_{cho}(P,\langle IN_{i}, OUT_{i}\rangle, R_i)$.

\begin{lemme}
\label{lem:Asperixcomp1}
Let $P$ be a normal logic program and $X$ be  an answer set of $P$.
Then there exists an \asperix. computation that converges to $X$.
\end{lemme}

\begin{proof} (of Lemma~\ref{lem:Asperixcomp1})
\label{proof:Asperixcomp1}
Let $P$ be a normal logic program and $X$ be  an answer set of $P$.
Then, there exists  an enumeration ${\langle r_i \rangle}_{i \in [1..n]}$ of $GR_P(X)$ that satisfies conditions (1) and (2) from Lemma \ref{lem:Asperixenumeration}.

Let  ${\langle R_i, \langle IN_i, OUT_i\rangle \rangle}_{i=0}^{\infty}$ be the sequence defined as follows.
\begin{itemize}
  \item $R_0 = \emptyset$, $IN_0 = \emptyset$ and $OUT_0 = \{\bot\}$
  \item $\forall i, 1\leq i\leq n$, $R_i = R_{i-1} \cup \{r_i\}$, $IN_i = IN_{i-1} \cup \{\HEAD.(r_i)\}$ and $OUT_i = OUT_{i-1}\cup \BODY.^-(r_i)$
  \item $\forall i > n$, $R_i = R_{i-1}$, $IN_i = IN_{i-1}$ and $OUT_i = OUT_{i-1}$
\end{itemize}
For all $i \in [1..n]$, we have:
\begin{enumerate}[label=(*\arabic*)]
\item $X =  \HEAD.(GR_P(X))$ (by Theorem \ref{the:AnswersetGR})
\item $IN_i = \bigcup_{j=1}^i \{\HEAD.(r_j)\}$ and $IN_{\infty} = \bigcup_{i=0}^{\infty}IN_i =  X$ (by (*1))
\item  $OUT_i = \bigcup_{j=1}^i \BODY.^-(r_j)$ and therefore $OUT_i \cap X = \emptyset$ (by Definition \ref{def:Generatingrules} of $GR_P(X)$)
\item $\Delta_{pro}(P,\langle IN_{i},OUT_{i} \rangle, R_{i}) \subseteq GR_P(X)$
\end{enumerate}
Property (*4) can be proved as follows.
By definition \ref{def:Delta_Cho_Pro}, $\Delta_{pro}^i = \{r \in  ground(P) \setminus R_{i} \mid \BODY.^+(r) \subseteq IN_{i} \mbox{ and } \BODY.^-(r) \subseteq OUT_{i}\}$. 
And by (*2) and (*3), $IN_{i} \subseteq X$ and $OUT_{i} \cap X = \emptyset$.
Thus $\Delta_{pro}^i \subseteq GR_P(X)$.

We are now able to prove that the sequence ${\langle R_i, \langle IN_i, OUT_i\rangle \rangle}_{i=0}^{\infty}$ is an \asperix. computation.

Let us first note that $\forall i,~\langle IN_i, OUT_i\rangle$ is a partial interpretation since $IN_i \cap OUT_i = \emptyset$ (by (*2) and (*3)).

Now we prove that Revision principle holds for each $i \geq 1$.
Let $i$ such that $1\leq i\leq n$, then $r_i$ is such that $\BODY.^+(r_i) \subseteq \HEAD.(\{r_k \mid k<i\}) = IN_{i-1}$. Two cases are possible.
First, if $\BODY.^-(r_i) \subseteq \BODY.^-(\{r_k \mid k<i\}) = OUT_{i-1}$, then $r_i \in \Delta_{pro}^{i-1}$ and Revision principle holds at rank $i$.
Second, if $\BODY.^-(r_i) \not\subseteq \BODY.^-(\{r_k \mid k<i\})$ then, by definition of enumeration 
${\langle r_i \rangle}_{i \in [1..n]}$, there is no rule $r_j$ with $j>i$ such that $\BODY.^+(r_j) \subseteq IN_{i-1}$ and $\BODY.^-(r_j) \subseteq OUT_{i-1}$. So $\Delta_{pro}^{i-1} \cap GR_P(X) = \emptyset$. 
And as $\Delta_{pro}^{i-1} \subseteq GR_P(X)$ (by (*4)), $\Delta_{pro}^{i-1} = \emptyset$.
Moreover, $r_i$ is a generating rule, thus $\BODY.^-(r_i) \cap X = \emptyset$ and $\BODY.^-(r_i) \cap IN_{i-1} = \emptyset$ (since $IN_{i-1} \subseteq X$).
Thereby $r_i \in \Delta_{cho}^{i-1}$ and Revision principle holds.
If $i > n$, Revision principle trivially  holds (Stability). 

At step $n$, we have $IN_{n} = \bigcup_{j=1}^n \{\HEAD.(r_j)\} = X$ and $R_{n} = \bigcup_{j=1}^n \{r_j\} = GR_P(X)$.
$\Delta_{cho}^{n+1} =  \{r \in  ground(P) \setminus R_{n} \mid \BODY.^+(r) \subseteq X \mbox{ and } \BODY.^-(r) \cap X = \emptyset\}$.
Thus $\Delta_{cho}^{n+1} = \emptyset$. 
Convergence principle holds and $IN_{\infty} = IN_{n} = X$.
\end{proof}

\begin{lemme}
\label{lem:Asperixcomp2}
Let $P$ be a normal logic program and ${\langle R_i, \langle IN_i, OUT_i\rangle \rangle}_{i=0}^{\infty}$ be an \asperix. computation for $P$.
Then,  $IN_{\infty}$ is an answer set of $P$.
\end{lemme}

\begin{proof} (of Lemma~\ref{lem:Asperixcomp2})
\label{proof:Asperixcomp2}
Let ${\langle R_i, \langle IN_i, OUT_i\rangle \rangle}_{i=0}^{\infty}$ be an \asperix. computation for $P$.

We first prove that $\forall i > 0,~\forall j \geq i-1,~R_{i}  \subseteq GR_P(IN_j)$.
For each rule $r_i$, $\BODY.^+(r_i) \subseteq IN_{i-1}$ and $IN$ set increases monotonically, thus $\BODY.^+(r_i) \subseteq IN_{j}, \forall j \geq i-1$.
If $r_i \in {\Delta}_{pro}^{i-1}$, then $\BODY.^-(r_i) \subseteq OUT_{i-1}$ and $OUT_{i-1} \cap IN_{i-1} = \emptyset$. 
Since $IN$ and $OUT$ sets grow monotonically with an empty intersection, $ \BODY.^-(r_i) \cap IN_j = \emptyset, \forall j \geq i-1$.
If $r_i \in {\Delta}_{cho}^{i-1}$, then $\BODY.^-(r_i) \cap IN_{i-1} = \emptyset$. 
And, since $OUT_i = OUT_{i-1} \cup \BODY.^-(r_i)$, we have $\forall j  \geq i,  \BODY.^-(r_i) \subseteq OUT_{j}$,
and thus, with the same reasonning as above ($r_i \in {\Delta}_{pro}^{i-1}$), $\BODY.^-(r_i) \cap IN_j = \emptyset, \forall j \geq i-1$.

 $R_i = \bigcup_{k=1}^{i}\{r_k\}$ and, since $\forall j \geq k-1, ~r_k \in  GR_P(IN_{j})$, $r_k \in  GR_P(IN_{i})$.
 Thus $R_i \subseteq GR_P(IN_i)$.

By Convergence principle we have $\exists i,~{\Delta}_{cho}^{i} =  \{r \in  ground(P) \setminus R_{i} \mid \BODY.^+(r) \subseteq IN_{i} \mbox{ and } \BODY.^-(r) \cap IN_{i} = \emptyset\} = \emptyset$, 
then $GR_P(IN_{i}) \subseteq R_{i}$.
Since $\forall i, R_i \subseteq GR_P(IN_{i})$, $GR_P(IN_{i}) = R_{i}$.
And  $IN_{i} = \HEAD.(R_{i})$ (by definition of an \asperix. computation), 
thus $IN_i = head(GR_P(IN_{i}))$.

Moreover, for all $i > 0$, $body^+(r_i) \subseteq IN_{i-1}$, thus $R_i = \bigcup_{k=1}^i\{r_k\}$ is grounded and, since $R_i = GR_P(IN_{i})$, $GR_P(IN_{i})$ is grounded.
Finally, by Theorem \ref{the:AnswersetGR}, $IN_{i}$ is an answer set of P.

\end{proof}

\begin{proof} (of Theorem \ref{the:Aspcomp})
Lemmas \ref{lem:Asperixcomp1} and \ref{lem:Asperixcomp2} prove each one direction of the equivalence.
\end{proof}

\subsection{Proof of Theorem \ref{the:mbtAsperixcomp}}
\label{subsec:proofmbtAspcomp}

Lemmas \ref{lem:mbtAsperixcomp1} and \ref{lem:mbtAsperixcomp2} establish completeness and correctness.

\begin{lemme}
\label{lem:mbtAsperixcomp1}
Let $P$ be a normal logic program and $X$ be an answer set for $P$.
Then there exists  a mbt \asperix. computation for $P$ that converges to X.
\end{lemme}

\begin{proof} (of Lemma~\ref{lem:mbtAsperixcomp1})
\label{proof:mbtAsperixcomp1}
Let $P$ be a normal logic program and $X$ an answer set for $P$.
By Theorem \ref{the:Aspcomp}, 
there exists an \asperix. computation ${\langle R_i, \langle IN_i, OUT_i\rangle \rangle}_{i=0}^{\infty}$ with $IN_{\infty} = X$.
Let $C = {\langle K_i, R_i, \langle IN_i, MBT_i, OUT_i\rangle \rangle}_{i=0}^{\infty}$ with $ K_i = MBT_i = \emptyset, ~\forall i \geq 0$.
$C$ is clearly a mbt \asperix. computation for $P$ where ``Rule exclusion'' is never used and thus ``Mbt-propagation'' is not used either.
\end{proof}

\begin{lemme}
\label{lem:mbtAsperixcomp2}
Let $P$ be a normal logic program and ${\langle K_i, R_i, \langle IN_i, MBT_i, OUT_i\rangle \rangle}_{i=0}^{\infty}$  be a mbt \asperix. computation for $P$.
Then $IN_{\infty}$ is an answer set of $P$.

\end{lemme}

\begin{proof} (of Lemma~\ref{lem:mbtAsperixcomp2})
\label{proof:mbtAsperixcomp2}
Let ${\langle K_i,  R_i, \langle  IN_i, MBT_i, OUT_i\rangle \rangle}_{i=0}^{\infty}$  a mbt \asperix. computation for $P$.
Then $C = {\langle R_i, \langle IN_i, OUT_i\rangle \rangle}_{i=0}^{\infty}$  is an \asperix. computation for $P$:
it satisfies Revision principles of  an \asperix. computation and
it trivially satisfies Convergence too.
By Theorem \ref{the:Aspcomp}, $C$ converges to an answer set $IN_{\infty}$.
\end{proof}

\begin{proof} (of Theorem \ref{the:mbtAsperixcomp})
Lemmas \ref{lem:mbtAsperixcomp1} and \ref{lem:mbtAsperixcomp2} prove each one direction of the equivalence.
\end{proof}

%

\end{document}